\documentclass{article}

\usepackage[margin=1in]{geometry}
\usepackage{graphicx}
\usepackage{amsmath}
\usepackage{natbib}
\usepackage{dsfont}

\usepackage[ruled,linesnumbered]{algorithm2e}
\SetEndCharOfAlgoLine{}

\usepackage{amsthm}
\usepackage{amsfonts}
\usepackage{multirow}

\usepackage{amsmath,bbm,graphicx,amssymb,amsthm,amsbsy}
\usepackage{xcolor,transparent}
\usepackage{hyperref}
\hypersetup{
    colorlinks=true,
    linkcolor=blue,
    filecolor=magenta,      
    urlcolor=cyan,
    pdfpagemode=FullScreen,
    }
\usepackage{tabularx,booktabs}
\usepackage[utf8]{inputenc}
\usepackage{url}
\usepackage{enumitem}
\setlist[enumerate]{nosep}
\setlist[itemize]{nosep}

\newcommand{\trace}[1]{\text{tr}\left(#1\right)}

\newtheorem{theorem}{Theorem}

\usepackage{setspace}
\usepackage[normalem]{ulem}

\makeatother

\newcommand{\sccomment}[1]{#1}

\newcommand\blfootnote[1]{%
  \begingroup
  \renewcommand\thefootnote{}\footnote{#1}%
  \addtocounter{footnote}{-1}%
  \endgroup
}

\begin{document}

\title{Principal component analysis balancing prediction and approximation accuracy {for spatial data}}
\author{Si Cheng$^{1}$ \and 
Magali N. Blanco$^2$ \and
Timothy V. Larson$^{2,3}$
Lianne Sheppard$^{1,2}$ \and
Adam Szpiro$^{1\star}$ \and
Ali Shojaie$^{1\star}$}
\date{
$^1$Department of Biostatistics, University of Washington\\
$^2$Department of Environmental \& Occupational Health Sciences, University of Washington\\
$^3$Department of Civil \& Environmental Engineering, University of Washington
}

\maketitle
\blfootnote{$^\star$ indicates joint senior authors}
\blfootnote{[funding information]. All authors declare no conflicts of interest.}

\begin{abstract}

{Dimension reduction is often the first step in statistical modeling or prediction of multivariate spatial data.} However, most existing dimension reduction techniques do not account for the spatial correlation {between observations and do not take the downstream} modeling task into consideration when finding the lower-dimensional representation. We formalize the closeness of {approximation to the original data and the utility of lower-dimensional scores} for downstream modeling as two {complementary, sometimes conflicting, metrics} for dimension reduction. We illustrate how existing methodologies fall into this framework and propose a flexible dimension reduction algorithm that achieves the optimal trade-off. We derive a computationally simple form for our algorithm and illustrate its performance through simulation studies, as well as two applications in air pollution modeling and spatial transcriptomics.

\end{abstract}

\section{Introduction}
\label{sec:intro}

Statistical modeling of multivariate data is a common task in many areas of research. Environmental epidemiologists often seek to learn the relationship between some health outcome and a complex mixture of multiple air pollutants, where the health effects of such mixture could vary depending on its chemical composition {rather} than the overall quantity alone \citep{goldberg2007interpretation, lippmann2013national,  dai2014associations, achilleos2017acute}. 
Spatial transcriptomics is another area involving the analysis of multivariate, and most often, high-dimensional measurements, where researchers analyze gene expressions on tissues with spatial localization information \citep{hu2021spagcn, zhao2021spatial}.

Despite the different domain areas, there are several common challenges for such tasks: First, the (potential) high dimensionality and strong autocorrelation of the multivariate measurements would require additional {exploratory data analysis}, 
such as dimension reduction, as a preliminary step; see e.g., \citet{dominici2003health} for epidemiological and \citet{kiselev2017sc3, sun2019accuracy} for genomics studies. 
The spatial characteristics of the measurements should be accounted for in {such}  dimension reduction. 

The second complexity is that dimension reduction is often executed independently from the subsequent modeling steps, which would cause the {dimension-reduced data} to be uninterpretable, or of sub-optimal {utility}, for downstream analyses. 
One example is the study of health effects of air pollution, where there is often spatial misalignment between the air pollution monitoring sites and locations where health outcomes are available \citep{ozkaynak2013air, bergen2013national, chan2015long}. Therefore, after lower-dimensional {scores} are obtained, they still need to be extrapolated to the latter locations; but these {lower-dimensional scores may be noisy and not spatially predictable.} 
This problem is also present in spatial gene expression analyses, where the reduced-dimensional gene expression matrix may not preserve all biologically meaningful patterns. When {the processed gene expression data} further undergo spatial imputation due to incomplete profiling \citep{pierson2015zifa, prabhakaran2016dirichlet, li2021imputation}, or are clustered into spatial domains based on spatial and other biological information (commonly termed \emph{spatial domain detection} in spatial transcriptomics) \citep{zhao2021spatial, long2023spatially}, the statistical performance may be affected by loss of information that seems unimportant in dimension reduction, but is significant for downstream modeling \citep{liu2022joint}.


Principal component analysis (PCA) is a classical dimension-reduction technique. {It leads} to independent {and} lower-dimensional scores, called principal component scores, or PC scores, {that approximate} the multivariate measurements \citep{jolliffe2002principal}. Existing methodologies that tackle some of the aforementioned challenges are often based on extensions of PCA.
\citet{jandarov2017novel} proposed a spatially predictive PCA algorithm, where the PC scores were constrained to be linear combinations of covariates and spatial basis, so they can be more accurately predicted at new locations. \citet{bose2018adaptive} extended this predictive PCA approach and enabled adaptive selection of covariates to be included for each PC. \citet{vu2020probabilistic, vu2022spatial} proposed a probabilistic version of predictive PCA along with a low-rank matrix completion algorithm, which {offers} improved performance when there is spatially informative missing data. \citet{keller2017covariate} developed a predictive $k$-means algorithm where dimension reduction was conducted by clustering.  

In spatial transcriptomics applications such as domain detection (see detailed discussion in Section~\ref{subsec:st}), \citet{zhao2021spatial} processed gene expression information {using} PCA and conducted downstream clustering via a Bayesian approach, where the spatial arrangement of the measured spots were modeled using the PC scores. \citet{shang2022spatially} proposed a probabilistic PCA algorithm, where spatial information was incorporated into a latent factor model for gene expression. \citet{liu2022joint} developed a joint approach that simultaneously {performs} dimension reduction and clustering, and uses a latent hidden Markov random field model to enforce spatial structures of the clusters.

When dimension reduction is used as an {initial} step before downstream analyses such as prediction, statistical inference and clustering, there are two key considerations for the quality of dimension reduction. The first is \textit{representability}, which is how well the lower dimensional components approximate the original measurements. {The second,} often conflicting, goal is  \textit{predictability} of the resulting components, so that they preserve meaningful scientific and spatial information {when extrapolated to locations without measurements}; {such predictability ensures that the extrapolations} are of high quality for subsequent modeling. {Even in analytical tasks that do not primarily focus on prediction (e.g., domain detection in spatial transcriptomics), the predictability of the PC scores is still desirable and implicitly considered, since it enforces the interpretability and spatial structure of the PCs; see Section~\ref{subsec:st} for further illustration.}

Among existing methodologies, {classical} PCA solely focuses on representability, while predictive PCA \citep{jandarov2017novel} and its variants \citep{bose2018adaptive} prioritize predictability and optimize representability only after the former is guaranteed. 
Probabilistic approaches, such as \citet{shang2022spatially}, \citet{liu2022joint} and \citet{vu2020probabilistic}, implicitly {incorporate} both tasks into the latent factor models, though their performance depends on the validity of the parametric model assumptions, and there is no explicit interpretation or optimality guarantee on either property or the trade-off between them. 

In this work, we propose a flexible dimension reduction algorithm, termed \emph{representative and predictive PCA} (RapPCA), that explicitly minimizes a combination of representation and prediction errors and finds the optimal balance between them. We further allow the underlying lower-dimensional scores to have complex spatial structure and non-linear relationships with {external} covariates (if any). We show that the optimization problem involved can be solved by eigendecomposition {of} transformed data, enabling simple and efficient computation.

We start in Section~\ref{sec:setting} by briefly reviewing related methods, introducing notations, and defining {various performance metrics of dimension reduction.} 
Section~\ref{sec:method} describes our proposed approach, and establishes the optimality of the proposed solution. We {compare} the performance of our method with existing variants of PCA via simulation studies in Section~\ref{sec:sims}, and demonstrate its application in epidemiological exposure assessment as well as spatial transcriptomics in Section~\ref{sec:data}. Section~\ref{sec:discussion} concludes the manuscript with a discussion of our methodology and related alternatives.





\section{Setting}
\label{sec:setting}

\subsection{{Classical} and Predictive PCA}
\label{subsec:review_pca}

{Let} $Y\in \mathbb R^{n\times p}$  represent $p$-dimensional spatially-structured variables of interest measured at $n$ locations. Examples include the concentrations of a mixture of $p$ pollutants at $n$ monitoring sites, or the expression of $p$ genes at $n$ spots on tissues. In addition, we observe the spatial coordinates $\{(s_{i1}, s_{i2})\}_{i=1}^n$ of the $n$ locations, along with (potential) observations of $d$ covariates, $X\in\mathbb R^{n\times d}$. These could be, for example, population density and land use information for air pollution studies, or histology information from images for spatial transcriptomics. 

{The} ultimate goal is to conduct prediction, clustering or other modeling tasks on $Y$ based on the spatial coordinates and covariates. Since the columns of $Y$ may be strongly correlated and/or noisy \citep{rao2021exploring}, dimension reduction {can be used to} to extract scientifically meaningful signals from the original data.
The {classical} 
PCA achieves this by maximizing the proportion of variability in $Y$ that is explained by the lower dimensional PCs.
We briefly introduce the {classical} PCA algorithm under the 
{formulation} of \citet{shen2008sparse} to highlight its connection with related techniques.

Suppose the columns of $Y$ are centered and scaled. To find an $r$-dimensional representation of $Y$, {classical} PCA can be expressed as a sequence of rank-1 optimization problems for $l = 1, \ldots, r$ {\citep{shen2008sparse}}:
\begin{equation}
    \min_{u_l,v_{l}}\ \left\lVert Y^{(l)} - u_l v_l^\top\right\rVert_F^2 \quad 
    \text{s.t. } \lVert v_l\rVert_2 = 1,
    \label{equ:pca}
\end{equation}
where the 1-dimensional PC score $u_l$ is an $n$-vector, the loading $v_l$ is $p \times 1$, and $\lVert\cdot\rVert_F$ and $\lVert \cdot\rVert_2$ represent Frobenius norm for matrices and $\ell_2$ norm for vectors, respectively. $Y^{(l)}$ is the residual from approximation after each iteration: denoting $\tilde u_l, \tilde v_l$ as the solution to (\ref{equ:pca}), we define $Y^{(l)} = Y^{(l-1)} - \tilde u_{l-1}\tilde v_{l-1}^\top$ for $2\le l\le r$ and $Y^{(1)} = Y$. 

{While} each of the $u_l$'s explain the greatest variability in $Y$, there is no guarantee that they are scientifically {interpretable, and they do not explicitly} account for the spatial structures underlying 
$Y$. 
{Recognizing this, p}redictive PCA \citep{jandarov2017novel} builds upon the expression in (\ref{equ:pca}), 
{but constrains each PC score $u_l$ to be in a model space of the spatial coordinates and covariates $X$, leading to the optimiztion problem}, 
\begin{equation}
    \min_{\alpha_l,v_l}\ \left\lVert Y^{(l)} - \left(\frac{Z\alpha_l}{\lVert Z\alpha_l\rVert_2}\right)v_l^\top \right\rVert_F^2
    \label{equ:predpca}
\end{equation}
for each $l = 1,\ldots, r$. Here, $Z = [X \ B]$, $B\in \mathbb R^{n\times m}$ contains $m$ thin-plate regression spline \sccomment{\citep[TPRS; ][]{wood2003thin}} basis functions capturing the spatial effects, {and in contrast to (\ref{equ:pca}), normalization is done on the term $Z{\alpha_l}$ instead of $v_l$}. 
{As before, $Y^{(l)}$ is the residual after approximation, i.e. $Y^{(l)} = Y^{(l-1)} - Y^{(l-1)}\frac{\tilde v_{l-1} \tilde v_{l-1}^\top}{\lVert \tilde v_{l-1}\rVert_2}$ for $2\le l\le r$ and $Y^{(1)} = Y$. {Note that the quantity $Y^{(l-1)}\tilde v_{l-1}$ {corresponds to} $\tilde u_{l-1}$ in the expression of PCA.}}
Compared {with} (\ref{equ:pca}), this formulation restricts each PC score {$u_l = \frac{Z\alpha_l}{\lVert Z\alpha_l\rVert_2}$} to fall within 
the space spanned by the columns of $Z$. {This {ensures} the spatial smoothness and the inclusion of information from covariates $X$ in the PC score $u_l$, and thus enforces predictability.}

After solving either (\ref{equ:pca}) or (\ref{equ:predpca}) for all $r$ PCs, the loadings are $\ell_2$-normalized ($\tilde v_r \leftarrow \tilde v_r/\lVert \tilde v_r\rVert_2$) if not already done
, and then concatenated as $\tilde V_{p\times r}:= [\tilde v_1 \ \ldots \ \tilde v_r]$. The PC scores are then defined as $\tilde U_{n\times r}:= Y\tilde V$.

As noted earlier, {classical} PCA aims to achieve optimal representability of the PC scores $\tilde u$, but these $\tilde u$'s may not retain {meaningful signals in $X$ or important spatial patterns to be well predictable.} In contrast, predictive PCA and {its variants} \citep{jandarov2017novel, bose2018adaptive} minimize the {approximation gap}, {which will be formally defined in Section~\ref{subsec:metrics}},  after constraining each $\tilde u$ to lie within a specific model space; but when these PC scores $\tilde u$ are predicted at new locations and transformed back to the higher dimensional space of $Y$, closeness in the space of $\tilde u$ may not necessarily translate to the original space of observations $Y$ due to reduced approximation accuracy. Our proposed method, {introduced in Section~\ref{sec:method}}, is based on an interpolation between {classical} and predictive PCA, and encourages $\tilde u$ to be close to, but not exactly within, the specified model space while balancing the quality of representation. 

\subsection{Evaluation Metrics for Dimension Reduction}
\label{subsec:metrics}

{Before introducing our proposal, it is helpful to formalize different aspects of dimension reduction performances, such as representability and predictability, {and to} characterize the overall balance between them. 

Suppose the PC scores $\tilde U_{\text{trn}} = [\tilde u_1 \ \ldots \tilde u_r]$ and loadings $\tilde V$ are obtained on a set of training data {with size $n_{\text{trn}}$}, and $\tilde u_1, \ldots,\tilde u_r$ are then predicted at {$n_{\text{tst}}$} test locations with a {user-specified} 
spatial prediction model as $\hat U_{\text{tst}}:=[\hat u_1 \ \ldots \ \hat u_r]$. We define the \textit{mean-squared prediction error} (MSPE) of this procedure as
\[
    \text{MSPE} = {n_{\text{tst}}}^{-1}\left\lVert (\hat U_{\text{tst}} - U^*) \tilde V^\top \right\rVert_F^2,
\]
where $U^* = \arg\min_U \lVert Y_{\text{tst}} - U\tilde V^\top\rVert_F^2 = Y_{\text{tst}}\tilde V$. {Thus, $U^*$} is what the actual PC scores on the test set would be according to the loadings $\tilde V$, if $Y_{\text{tst}}$ were known. Consequently, MSPE characterizes the gap between the predicted and true PC scores and reflects the predictability of the PCs.

The \textit{mean-squared representation error} (MSRE) is defined as
\[
    \text{MSRE} = {n_{\text{tst}}}^{-1}\left\lVert Y_{\text{tst}} - U^*\tilde V^\top \right\rVert_F^2,
\]
which measures the gap {between the true data, $Y_{\text{tst}}$, and} the closest approximation possible based on $\tilde V$. Since the quality of representation alone, without considering predictive performance, is of more interest for the training than the test data, we instead examine MSRE-trn as the metric for representation: $\text{MSRE-trn} = {n_{\text{trn}}}^{-1}\lVert Y_{\text{trn}} - \tilde U_{\text{trn}}\tilde V^\top\rVert_F^2$. It can be seen that MSRE-trn exactly matches the objective function for PCA in (\ref{equ:pca}).

{Given any outcome data $Y$, loadings $\tilde V$ from a dimension reduction procedure and predicted PC scores $\hat U$, we define the \textit{total mean-squared error} (TMSE) resulting from both dimension reduction and prediction as $\text{TMSE}= n^{-1}\left\lVert Y - \hat U \tilde V^\top \right\rVert_F^2$}. This quantity measures the discrepancy between the true data $Y$ and the predicted scores that are transformed back to the original, higher-dimensional space.
In the specific training-test setting, the expression of TMSE becomes
\[
    \text{TMSE} = {n_{\text{tst}}}^{-1}\left\lVert Y_{\text{tst}} - \hat U_{\text{tst}}\tilde V^\top \right\rVert_F^2.
\]
MSPE and MSRE {can be viewed informally as} a decomposition of the overall TMSE, {where the prediction and approximation gaps, i.e., $(U^* - \hat U_{\text{tst}}) \tilde V^\top$ and $Y_{\text{tst}} - U^*\tilde V^\top $, {are components of} the overall error $Y_{\text{tst}} - \hat U_{\text{tst}}\tilde V^\top$.} 



\section{\sccomment{The Proposed Method}}
\label{sec:method}

\subsection{\sccomment{Representative and Predictive PCA (RapPCA)}}
\label{subsec:method-1}

Predictive PCA restricts the PC scores $\tilde u$ within some model space --- specifically, $\text{span}(Z)$ --- to enforce predictability while trading off representability to some extent. A natural idea, with more flexibility, would be relaxing this constraint, and instead {using a penalty term to force $\tilde u$ to be close} to, but not exactly within, the model space. By choosing the magnitude of {such a penalty term} data-adaptively, 
we aim to achieve the optimal balance between representation and prediction leading to the smallest overall error. 

{Noting} that the model space can be generalized from the linear span in (\ref{equ:predpca}) to incorporate more flexible mechanisms underlying each PC $\tilde u$, we adopt kernel methods \citep{evgeniou2000regularization, hastie2009elements} to capture non-linear covariate effects; {additionally}, we use smoothing splines \citep{wahba1990spline, wood2003thin} {rather than unpenalized regression splines} to enforce spatial structures. 

{More conceretely}, we propose to solve the following sequence of optimization problems, for $l = 1,\ldots,r$, to extract $r$ PCs from $Y$:
\begin{align}
    \notag &\sccomment{\min_{v,\alpha,\beta}\ f_{\gamma,\lambda_1,\lambda_2}(v,\alpha,\beta):= \left\lVert Y^{(l)}-Y^{(l)}vv^\top \right\rVert_F^2 +
    \gamma\left\lVert Y^{(l)}v-(K\alpha+B\beta) \right\rVert_2^2}
    +\lambda_1 \alpha^\top K\alpha + \lambda_2 \beta^\top Q\beta \\
    &\text{s.t. } 
    v^\top v = 1.
    \label{equ:rappca}
\end{align}

As in Section~\ref{sec:setting}, $Y^{(l)}$ {denotes} 
the residual $Y^{(l-1)}-\tilde u_{l-1}{\tilde v}_{l-1}^\top$ for $2\le l\le r$, \sccomment{where $\tilde u_{l-1} = Y^{(l-1)}\tilde v_{(l-1)}$}, and \sccomment{$Y^{(l)}$} {is} set to {the original data} $Y$ when $l = 1$.
Moreover, $K{\in\mathbb R^{n\times n}}$ is a kernel matrix such that $K_{ij} = k(X_{i\cdot}, X_{j\cdot})$ for some kernel function $k(\cdot, \cdot)$, with $X_{i\cdot}$ being the $i$th observation (row) of $X$. 
{The kernel function $k(\cdot,\cdot)$ maps the covariates $X$ into a higher-dimensional space, and the kernel matrix $K$ encodes the pairwise similarities between each entry $X_{i\cdot}$ in this transformed space. Linear operations in the high-dimensional space effectively captures non-linear relationships in the original input space of $X$. This enables more flexibility than most existing predictive variants of PCA which specify linear relationships \citep[e.g.,][]{jandarov2017novel, bose2018adaptive, vu2020probabilistic, shang2022spatially}. Common choices of kernel functions include linear kernel, $k(x,x')=x^\top x'$ \sccomment{and quadratic kernel $k(x,x')=(1+x^\top x')^2$. Other choices, such as Gaussian kernel, $k_h(x,x')=\exp\left(-h\lVert x-x' \rVert^2\right)$, or higher-order polynomial kernel $k_h(x, x') = (1+x^\top x')^h$ may include an additional tuning parameter, $h$} \citep{hofmann2008kernel, murphy2012machine}.}
The columns of $B{\in\mathbb R^{n\times m}}$ are evaluations of $m$ spline basis functions, {e.g., thin-plate regression splines as in predictive PCA,} at the spatial coordinates $\{(s_{i1}, s_{i2})\}_{i=1}^n$, and $Q{\in\mathbb R^{m\times m}}$ is the penalty matrix {induced by the choice of smoothing splines. For example, the penalty $Q$ corresponding to thin-plate regression splines is the wiggliness
penalty defined in \citet{wood2003thin}; for P-splines, the penalty term depends on the differences between consecutive entries of $\beta$, i.e., $\sum_k (\beta_k - \beta_{k-1})^2$} \citep{wood2017generalized}.

\sccomment{Problem~(\ref{equ:rappca}) is equivalent to optimizing 
$$
\tilde f_{\gamma,\lambda_1,\lambda_2}(u,v,\alpha,\beta):= \left\lVert Y^{(l)}-uv^\top \right\rVert_F^2 + \gamma\left\lVert u-(K\alpha+B\beta) \right\rVert_2^2  +\lambda_1 \alpha^\top K\alpha + \lambda_2 \beta^\top Q\beta,$$ 
subject to an additional constraint $u = Y^{(l)}v$. Its} first term $\left\lVert Y^{(l)}-uv^\top \right\rVert_F^2$ measures the approximation gap of the score $u$ and loading $v$, similar to the objective of minimization for {classical} PCA. The second term reflects the distance between the score $u$ and the model space specified by $K$ and $B$, and encourages $u$ to be predictable based on the covariates and spatial coordinates. 
The \sccomment{implicit} constraint $u = Y^{(l)}v$ ensures a similar relationship between the PC scores and loadings as in classical PCA. 

The computation of the optimization problem in (\ref{equ:rappca}) could suffer when $K$ or $Q$ is near-singular. Therefore, in practice, we replace $K$ and $Q$ \sccomment{in the penalty term} with $\tilde K := K+\delta I_n$ and $\tilde Q:= Q+\delta I_m$ for a small, data-independent constant $\delta$, and identity matrices $I_n$ and $I_m$, to avoid near-singularity; {see also \citet{li2019prediction} for similar handling}. 

\sccomment{Among the tuning parameters $\gamma$, $\lambda_1$ and $\lambda_2$}, 
$\gamma$ is of most interest as it controls the trade-off between predictability and representability, {with} larger $\gamma$ imposing higher penalty on unpredictable scores {$u= Y^{(l)}v$}. The third and fourth terms in (\ref{equ:rappca}) are $\ell_2$ penalties {that enforce the identifiability of the model parameters and help to avoid overfitting the data}. 


Though optimizing $v,\alpha,\beta$ iteratively via coordinate descent is straightforward {to implement}, the optimization problem (\ref{equ:rappca}) is {non-convex due to the joint optimization of $v$ and $(\alpha,\beta)$, as well as the non-convex feasible set of $v$. Consequently, iterative algorithms, including coordinate descent are not guaranteed to converge to the global minimizer.}
Instead, we propose an analytical solution to (\ref{equ:rappca}) that {attains the global miminum despite} the non-convexity of the objective function $f_{\gamma,\lambda_1,\lambda_2}(v,\alpha,\beta)$ {and/or the constraints}. 
We describe {our proposed procedure in Algorithm~\ref{algo:rappca}, and} refer to the proposed method as \textit{representative and predictive PCA}, or RapPCA. {Theorem~\ref{thm:optim}, proved in Appendix~\ref{app:proof}, establishes the optimality of $\tilde U$ and $\tilde V$}. {Numerical examples in Appendix~\ref{app:optimality} empirically verify the optimality}.

\RestyleAlgo{boxruled}
\begin{algorithm}[t]
  \caption{Representative and Predictive PCA (RapPCA) \label{algo:rappca}}
  \KwData{Outcome $Y^{n\times p}$; spatial coordinates $\{(s_{i1}, s_{i2})\}$; covariates $X^{n\times d}$ (optional)}
  \KwIn{Number of PCs to extract $r > 0$; hyper-parameters $\gamma^{(l)}$, $\lambda_1^{(l)}$, $\lambda_2^{(l)}$ specified for the $l$th PC ($l = 1,\ldots, r$); kernel function $k(\cdot, \cdot)$; spatial spline basis evaluated at $\{(s_{i1}, s_{i2})\}$: $B^{n\times m}$; penalty matrix $Q_{m\times m}$; $\delta > 0$}
  
  Calculate the kernel matrix $K$: $K_{ij}\leftarrow k(X_{i\cdot}, X_{j\cdot})$\;

  Add a small diagonal term to avoid near-singularity: $\tilde K \leftarrow K + \delta I_{n}$, $\tilde Q\leftarrow Q + \delta I_m$\;
  
  Initialize $Y^{(0)} = Y$, $\tilde u_0 = 0$, $\tilde v_0 = 0$\;
        
    \For{$l = 1, \ldots, r$}{
        $Y^{(l)} \leftarrow Y^{(l-1)} - \tilde u_{l-1}\tilde v_{l-1}^\top $\;

        Calculate the singular value decomposition of $Y^{(l)}$: $Y^{(l)} = \mathcal S^{(l)}{\mathcal D^{(l)}} {\mathcal T^{(l)}}^\top$\;

        Concatenate the columns of $K$ and $B$: $Z \leftarrow \left[K \quad \sqrt{\frac{\lambda_2^{(l)}}{\lambda_1^{(l)}}}B\right] $\;

        Concatenate the penalty matrices: $P \leftarrow \begin{bmatrix}
        \tilde K & 0 \\ 0 & \frac{\lambda_2^{(l)}}{\lambda_1^{(l)}}\tilde Q
        \end{bmatrix}$\;

        Set $A\leftarrow -(\gamma^{(l)}-1){\mathcal D^{(l)}}^2 + {\gamma^{(l)}}^2{\mathcal D^{(l)}}{\mathcal S^{(l)}}^\top Z(\gamma^{(l)} Z^\top Z + \lambda_1^{(l)} P)^{-1}Z^\top \mathcal S^{(l)}{\mathcal D^{(l)}}$\;
        
        Calculate the loading: $\tilde v_l \leftarrow {\mathcal T^{(l)}}^\top q$, where $q$ is the normalized first eigenvector of $A$\;

        Calculate the $l$th PC score: $\tilde u_l \leftarrow Y^{(l)}\tilde v_{l}$\;
    }

    Concatenate the PC scores and loadings: $\tilde U \leftarrow [\tilde u_1, \ldots, \tilde u_r]$, $\tilde V \leftarrow [\tilde v_1, \ldots, \tilde v_r]$\;
    
    \KwOut{PC scores $\tilde U$ and loadings $\tilde V$}
\end{algorithm}

\begin{theorem}
    For each PC (indexed by $l$), denote the singular value decomposition of $Y^{(l)}$ as $Y^{(l)} = \mathcal S^{(l)}{\mathcal D^{(l)}} {\mathcal T^{(l)}}^\top$. Take $Z,P,q$ and {the loading} $\tilde v_l$ to be the specified forms in Algorithm~\ref{algo:rappca}. Then $(\tilde v_l, \tilde\alpha_l,\tilde\beta_l)$ is the global minimizer of (\ref{equ:rappca}), where $\left[\tilde\alpha_l^\top \quad \sqrt{\frac{\lambda_2}{\lambda_1}}\tilde\beta_l^\top\right]^\top = (\gamma Z^\top Z + \lambda_1 P)^{-1}(\gamma Z^\top \mathcal S^{(l)}{\mathcal D^{(l)}}q)$.
\label{thm:optim}
\end{theorem}


\subsection{\sccomment{Computational Considerations}}
\label{subsec:method-2}

\sccomment{When $m$ spline basis functions are included in $B$ to represent a $p$-dimensional $Y$ observed over $n$ samples, the computational complexity of the procedure described in Algorithm~\ref{algo:rappca} is $O\left((n+m)^3 + np^2\right)$.} 
This algorithm is computationally simple in that it obtains the optimal solution with an explicit expression, as opposed to iterative numerical optimization procedures.

The proposed algorithm depends on the specification of several quantities. For the hyper-parameters $\gamma^{(l)},\lambda_1^{(l)},\lambda_2^{(l)}$, one could adopt cross-validation and choose the combination that minimizes TMSE or other metrics (e.g., MSPE/MSRE) depending on the targeted use case, sequentially for each PC.
Taking TMSE as an example, for the $l$th PC and for each combination of $\gamma^{(l)},\lambda_1^{(l)},\lambda_2^{(l)}$, we calculate $\tilde u_l, \tilde v_l$ on the training set, predict $\hat u_{\text{tst}}$ on the validation set, and then examine $\text{TMSE}_l = n_{\text{tst}}^{-1}\left\lVert Y^{(l)}_{\text{tst}} - \hat u_{\text{tst}}\tilde v_l^\top \right\rVert_2^2$ to choose the optimal combination minimizing $\text{TMSE}_l$.
\sccomment{When necessary, the additional tuning parameter for kernel function, such as $h$ for the Gaussian kernel, can be chosen together with $\gamma, \lambda_1,\lambda_2$ using cross-validation.}
Specification of the dimension $m$ of $B$ is not critical 
due to the penalty on $\beta$ \citep{wood2017generalized}. One can start with a large $m$ (e.g., close to the sample size $n$) and the degree of freedom of $\beta$ will be controlled data-adaptively via $\lambda_2$.

The \sccomment{implementation of} Algorithm~\ref{algo:rappca} \sccomment{also} depends on the number $r$ of PCs. Since increasing $r$ will reduce the total representation error and increase the total prediction error, the selection can be done by examining the trajectories of MSPE, or cumulative sum of PC-wise prediction errors, along with MSRE with respect to $r$, and identifying the `elbow' of these curves. This reflects the threshold beyond which adding more PCs would not alter either the prediction or representation meaningfully. \sccomment{We will present an illustrative example (Figure~\ref{fig:elbow_curves}) for this procedure in Section~\ref{subsec:st}, which involves a high-dimensional breast tumor dataset.}

\section{Simulations}
\label{sec:sims}

In this section, we {compare} the performance of RapPCA {on} simulated data {with} {classical} and predictive PCA. In particular, we investigate three scenarios with different trade-offs between predictability and representability resulting from different data generating mechanisms. 

In all settings, we randomly generate $n = 200$ locations $\{(s_{i1}, s_{i2})\}_{i=1}^n$ within $[0,1] \times [0,1]$. For each location, we generate $d = 10$ covariates $x_{ij}$ ($i = 1,\ldots, n$ and $j = 1, \ldots, d$) from independent Uniform$[-1,1]$ distributions and calculate 6 independently distributed PCs as
\begin{equation}
    \text{PC}^{(l)} = f_l(X) + \epsilon_l,
    \  l = 1,\ldots, 6,
    \label{equ:sim-pc-defn}
\end{equation}
where $\epsilon_l\sim \text{Normal}(0, \sigma_l^2\Sigma)$ with $\sigma_l^2$ controlling the signal-to-noise ratio and hence the predictability of each PC; $f_l(\cdot)$ is a specified mean function, and the covariance matrix $\Sigma$ has an exponential structure with $\Sigma_{ii'} = 0.5\exp\left(-\frac{(s_{i1} - s_{i'1})^2 + (s_{i2} - s_{i'2})^2}{0.5}\right) + 0.5$. The outcome $Y$ represents concentrations of $p=15$ pollutants, and is given by
\begin{equation*}
    Y_{n\times p} = \text{PC}_{n\times 6} M_{6\times p} + \varepsilon_{n\times p} \, , 
\end{equation*}
where $M$ is adjusted differently in each scenario to control the contribution from different PCs, 
and the entries of the noise terms $\varepsilon$ are drawn from independent Normal$(0,0.1)$ distributions. 

We examine metrics including TMSE, MSPE and MSRE-trn for 100 replicates of data, in terms of their overall magnitudes, as well as the breakdown by each PC and/or $\gamma$ (the tuning parameter). 
We also report the individual prediction MSEs for each PC, which is $\text{MSE}_l = n_{\text{tst}}^{-1}\lVert \hat u_l - u^*_l \rVert_2^2$, recalling from Section~\ref{subsec:metrics} that $u^*_l = Y_{\text{tst}}\tilde v_l$. We extract $r=3$ PCs sequentially with each of the three candidate algorithms, PCA, predictive PCA and RapPCA. The optimal {combination of} $\gamma$, {$\lambda_1$ and $\lambda_2$} is selected by minimizing TMSE via {10-fold} cross-validation \sccomment{on a grid where $\gamma, \lambda_1$ and the ratio $\lambda_2/\lambda_1$ each independently varies from 0.05 to 5 with varying spacings\footnote{\sccomment{We use 0.05 for hyperparameter values up to 0.1; 0.1 for hyperparameter values between 0.1 and 1; 1 for hyperparameter values over 1.}}, plus the combination $\gamma = \lambda_1 = \lambda_2 = 0$ representing classical PCA as a special case}. 
\sccomment{We illustrate the performance of RapPCA with polynomial kernel $k_h(x, x') = (1+x^\top x')^h$, with $h = 1$ (linear) for scenarios 1 and 2, and $h=2$ for scenario 3. We adopt thin-plate regression splines implemented by the \texttt{mgcv} R package to capture the spatial pattern in each PC, with dimension $m$ equal to the size of training data for RapPCA, and the penalty term is chosen as cubic spline penalty with second order derivatives. Since predictive PCA does not directly handle high-dimensionality, we truncate the spline basis terms $B$ at $m = 10$ for predictive PCA.
The value of $\delta$ is fixed at 0.05, which is small relative to $\gamma,\lambda_1,\lambda_2$ but helps avoid near-singular penalties.} 
{For each PC score $\tilde u_1, \tilde u_2, \tilde u_3$, test set predictions are obtained via a two-step procedure {\citep[see, e.g.,][]{wai2020random}}, 
where we first train a random forest model, 
and then conduct spatial smoothing on the residuals with thin-plate regression splines.}
{We discuss the role played by the primary tuning parameter $\gamma$ in this section. Figure~\ref{fig:sim-lambdas} in Appendix~\ref{app:more-sim-rslts} provides further details about the impact of $\lambda_1, \lambda_2$ on the solution of RapPCA}. 

\subsection{Scenario 1: PCs with Equal Contribution}
\label{subsec:sim-case1}
We first consider the setting with 3 well-predictable PCs and 3 PCs mainly consisting of structured and unstructured Gaussian errors. More specifically, {for the mean model in (\ref{equ:sim-pc-defn})}, we let $f_l(X) = X\beta_l$ where the entries of $\beta_l$ are drawn independently from Uniform$[-1,1]$, and $\sigma_l^2 = 0$, for $l = 1,2,3$; we let $f_l = 0$ and $\sigma_l^2 = 1$ for $l = 4,5,6$. $M$ is set to be $\Lambda \tilde M$, where $\Lambda$ is diagonal with $\Lambda_{ll} = 1/\text{sd}\big(\text{PC}^{(l)}\big)$ and the entries of $\tilde M$ are independent Uniform$[-1,1]$. In other words, we first scale the 6 PCs to have equal variability, and then assign random but overall comparable weights to them. 

{Comparing the two existing methods,} it is expected that in this scenario {predictive} PCA would outperform {classical} PCA in prediction without severely compromising approximation accuracy, as the predictable PCs explain a similar amount of variability compared to the unpredictable ones. Consequently, predictive PCA would achieve better overall performance as reflected by TMSE. RapPCA flexibly interpolates between {classical} and predictive PCA, and is expected to have comparable or better performance {than predictive PCA}.

Panel A in Figure~\ref{fig:sims-gm} presents the breakdown of MSPE, MSRE-trn 
and TMSE for the first PC by $\gamma$, in comparison to PCA and predictive PCA for this scenario. We only present this breakdown for the first PC but not the subsequent ones because {starting from the second PC}, different PCA algorithms will have different residuals $Y^{(l)}$ (recall Equations~\ref{equ:pca},~\ref{equ:predpca} and~\ref{equ:rappca}), and the PC-specific breakdown of metrics is consequently no longer comparable. We observe that MSPE decreases and training MSRE increases as we increase $\gamma$ for RapPCA, since this imposes a greater penalty on prediction errors in optimization. 

Figure~\ref{fig:sims-overall} compares the prediction MSEs for each PC, as well as the overall TMSE, MSPE and MSRE-trn for different methods. 
As expected for this scenario, we observe in Panel A that predictive PCA achieves lower prediction errors than PCA for each PC, as well as {lower} overall MSPE. {This} advantage does not cost a higher approximation gap (MSRE), since the predictable components in the outcome $Y$ explain a similar amount of variability as the unpredictable ones. {RapPCA} 
is able to adjust the penalties on the covariate versus spatial effect terms adaptively, especially when there are a large number of spatial basis terms (recall the separate tuning parameters $\lambda_1$ and $\lambda_2$ in Equation~\ref{equ:rappca}), as opposed to predictive PCA which is restricted to a low-dimensional combination of covariates and spatial basis functions (recall the term $Z$ in Equation~\ref{equ:predpca}). Consequently, predictive PCA achieves better overall performance (reflected by smaller TMSE) than PCA, while RapPCA further improves the prediction and overall performance due to its flexibility to {data-adaptively} capture the covariate and spatial effects.

\begin{figure}
    \centering
    \includegraphics[width=16.5cm]{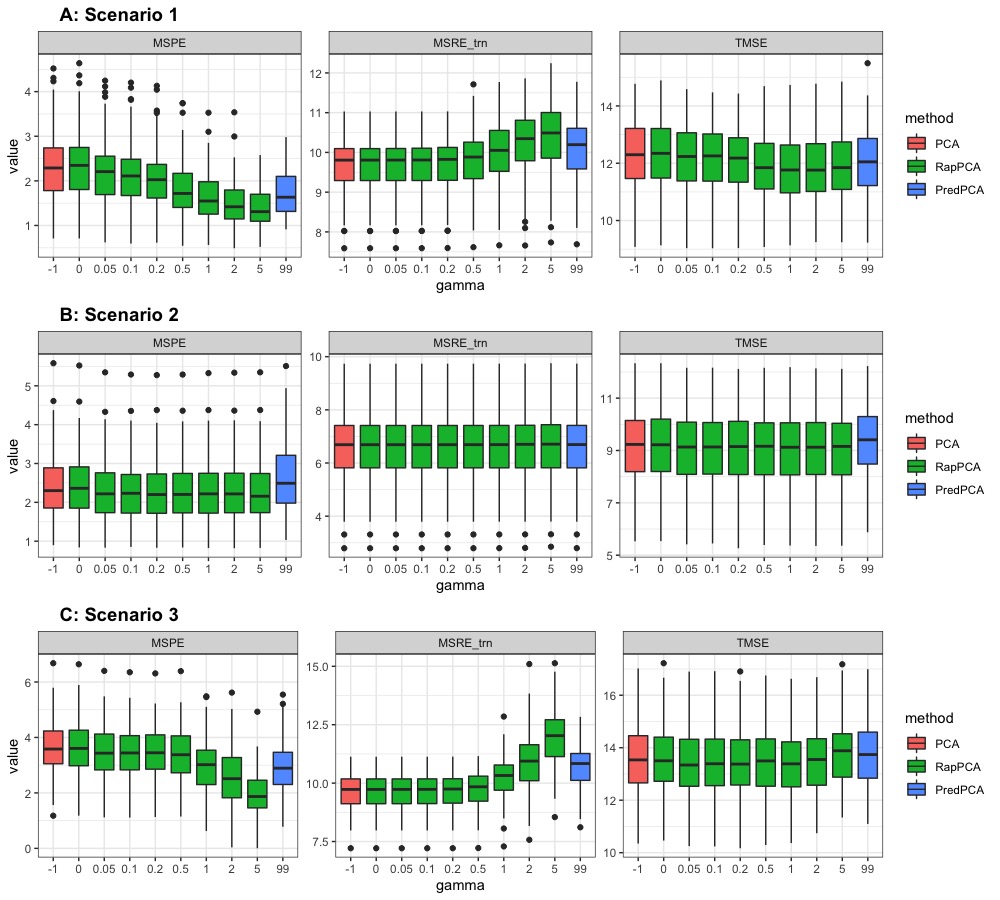}
    \caption{Breakdown of MSPE, MSRE-trn (MSRE on the training set) and TMSE for the first PC by $\gamma$ arcoss 100 replicates of data, with classical PCA (coded as $\gamma = -1$) and predictive PCA (\texttt{PredPCA}, coded as $\gamma = 99$) results presented for reference.}
    \label{fig:sims-gm}
\end{figure}

\begin{figure}
    \centering
    \includegraphics[width=16.5cm]{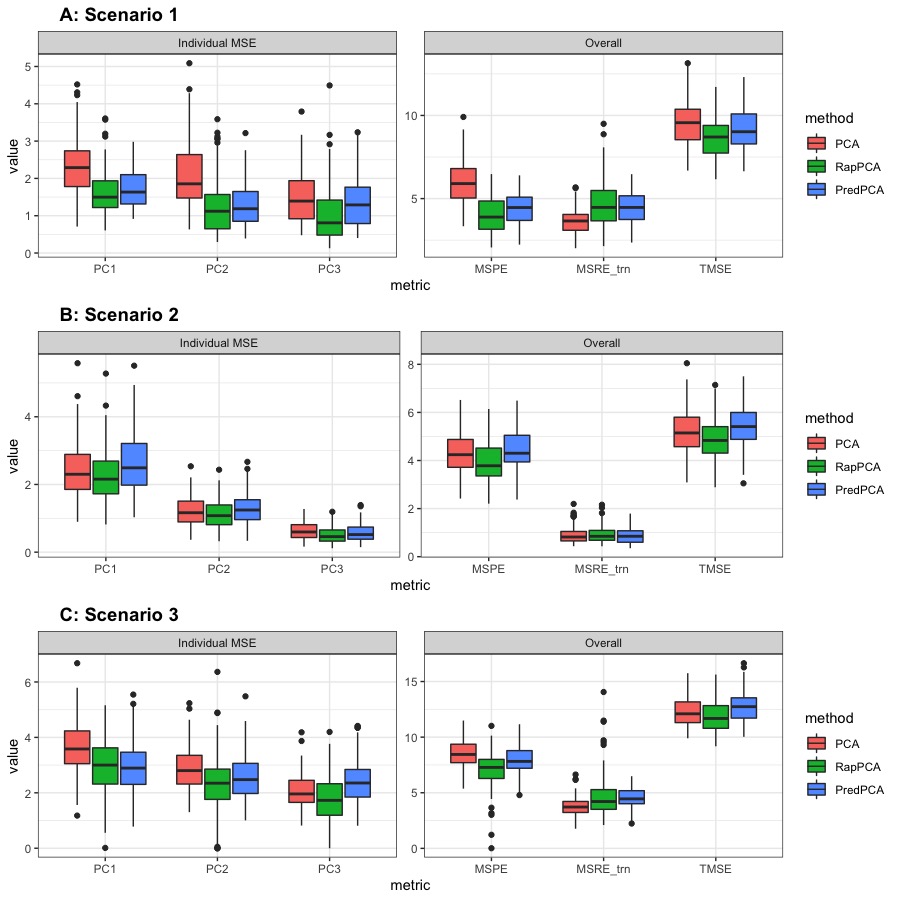}
    \caption{Distributions of individual prediction MSEs for each PC (left) and overall metrics for all PCs (right) across 100 replicates of data. {The optimal tuning parameters $\gamma,\lambda_1$ and $\lambda_2$ are selected separately for each simulation by minimizing TMSE via cross-validation. 
    }}
    \label{fig:sims-overall}
\end{figure}

\subsection{Scenario 2: PCs with Unequal Contribution}
We now consider a scenario where the predictable PCs explain a higher proportion of variability than the unpredictable PCs. In particular, the simulation setting is similar to Scenario 1, except that the entries of $\tilde M$ are first drawn from independent Uniform$[-1,1]$ distributions, and then the rows are scaled to have decaying norm (recall that the first 3 PCs are more predictable than the last 3 PCs). We expect in this case that PCA and predictive PCA would have comparable 
prediction, approximation and overall accuracy, since the well-predictable components of $Y$ coincide with those that explain the greatest amount of variance.

For the same reason, adjusting $\gamma$ in this case would not lead to steep changes in MSPE, MSRE or TMSE for {RapPCA}, as reflected by the flat curves of each metric in Panel B of Figure~\ref{fig:sims-gm}. Despite the similar behaviors of PCA and predictive PCA, Panel B of Figure~\ref{fig:sims-overall} reveals that RapPCA is still able to improve the prediction and overall performance by more flexibly exploiting the covariate and spatial information when identifying the PCs. 

\subsection{Scenario 3: PCs with Non-Linear Mean}
{In our last} simulation, we investigate the setting where the relationship between the true PCs and the covariates is non-linear. We modify the setting in Scenario 1 so that 
\[
    f_l(X) = X^{\odot 2}\beta_l + 2\sum_{j=1}^{\lfloor d/2 \rfloor} (X_{2j-1}\odot X_{2j}) \alpha_{lj}, 
\]
for $l = 1,2,3$, where $\odot$ and $\odot^ 2$ denote element-wise multiplication and square, and $\lfloor x\rfloor$ denotes the largest integer not exceeding $x$. In other words, we specify the mean function to be a combination of squared and interaction terms of the covariates, where interactions exist between pairs of consecutive covariates ($X_{\cdot 1}$ and $X_{\cdot 2}$; $X_{\cdot 3}$ and $X_{\cdot 4}$, etc). 

Panel C of Figure~\ref{fig:sims-gm} indicates that {by flexibly} handling non-linear effects, {RapPCA gains} a clear advantage in prediction accuracy compared {to} both predictive and {classical} PCA in this scenario, when we impose a large enough penalty $\gamma$ on prediction errors in optimization. It is not surprising that putting an emphasis on prediction leads to some loss in approximation accuracy; however, we observe from Panel C of Figure~\ref{fig:sims-overall} that by {balancing the two objectives} data-adaptively, RapPCA achieves meaningful improvement in MSPE while maintaining comparable, if not superior, overall performance.

\section{Applications}
\label{sec:data}

We illustrate the utility of RapPCA {using} real datasets in two different domain areas --- environmental epidemiology and spatial transcriptomics. The first application of dimension reduction in Section~\ref{subsec:trap} {($n = 309,\ p = 23$)} represents the case where practitioners seek to improve predictability (MSPE) and the overall performance (TMSE), and hence RapPCA is able to explicitly optimize these metrics as desired. 
The second example in Section~\ref{subsec:st} {($n = 607,\ p = 10,053$)} includes the scenario where tasks other than spatial prediction (e.g.,  clustering) {are} of interest; we demonstrate that RapPCA, though not directly incorporating the ultimate goal into its optimization, is able to capture meaningful biological and spatial information in the extracted PCs by enforcing predictability, {leading to} desirable model performance.

\subsection{Analysis of Seattle Traffic-Related Air Pollution Data}
\label{subsec:trap}

The performance of RapPCA in comparison with common existing methods, including {classical} and predictive PCA, is first illustrated with the multivariate traffic-related air pollution (TRAP) data in Seattle \citep{blanco2022characterization}. The study {leverages} a mobile monitoring campaign where a vehicle equipped with air monitors repeatedly collected two-minute samples at $n=309$ stationary roadside sites in the greater Seattle area. Repeated samples for the concentrations of 6 types of pollutants were obtained, including ultrafine particles (UFP), black carbon (BC), nitrogen dioxide (NO\textsubscript{2}), {carbon monoxide} (CO), carbon dioxide (CO\textsubscript{2}) and fine particulate matter (PM\textsubscript{2.5}). {Following \citet{blanco2022characterization}, we} focus on UFP, BC and NO\textsubscript{2} in our analysis {due to data quality considerations}. For UFP and BC, the concentrations were measured with multiple instruments corresponding to different measurement ranges or units. In particular, {NanoScan instrument with 13 different bin sizes}, as well as DiSCmini and PTRAK 8525 (with and without diffusion screen) instruments measured the concentration of UFP in terms of the counts of particles with different {size ranges}, along with the median size and overall count of particles.
Black carbon was assessed by micro-aethalometer (MA200), which {measures} the concentration of particles with different light absorbing properties, corresponding to 5 different ranges of wavelengths.

The median 2-minute visit concentrations were winsorized at the site level such that concentrations below the 5\textsuperscript{th} and above the 95\textsuperscript{th} quantile for a given site were set to those thresholds, respectively. This was done to reduce the influence of outliers in the measurements. These winsorized medians were then averaged for each site, leading to {$p=23$} annual average measurements of pollutant concentrations in total, including {17} for UFP, 5 for BC, and one for NO\textsubscript{2}. 
We use the concentration of NO\textsubscript{2} to normalize all other variables except {the median size measurements for UFP from the DiSCmini}. All variables are centered and scaled before running each PCA algorithm. 

We are interested in extrapolating the pollutant measurements to unobserved locations in the same study region. Because of the autocorrelation between these measurements, predicting each of them separately {might not always} lead to sensible results; instead, dimension reduction is needed before building spatial prediction models. {Dimension reduction also allows us to understand the health effects of pollutant mixtures rather than separate pollutants alone.} As in Section~\ref{sec:sims}, we compare the performance of RapPCA with {classical} and predictive PCA, each extracting the top 3 PCs from the {23} measurements of air pollutants. 
Due to the high dimensionality of geographical covariates, we first ran PCA on these 189 covariates and used the top 15 PCs as predictors for predictive PCA. {PCA and RapPCA were applied without this additional step}. \sccomment{We chose a linear kernel and the same configurations for thin-plate regression splines as in Section~\ref{sec:sims}. We adopted the same grid for the tuning of hyperparameters ($\lambda_1,\lambda_2/\lambda_1$) from 0.05 to 5 and the same value $\delta = 0.05$, but expanded the candidate range of $\gamma$ to $[0.05, 50]$ to accommodate wider ranges of prediction and representation errors}.
We examine the overall metrics as well as the individual prediction MSEs for each PC via 10-fold cross-validation {following} the same spatial prediction approach as Section~\ref{sec:sims}, namely, a random forest model followed by spatial smoothing via TPRS on the residuals {\citep{wai2020random}}.

Table~\ref{tab:trap-cv} compares each dimension reduction algorithm in terms of cross-validated TMSE, MSPE and MSRE-trn as well as prediction MSEs for each PC. Consistent with the findings from simulation studies, RapPCA achieves the lowest prediction and overall errors, while PCA has the smallest approximation gap on the training data. The advantage in predictability of RapPCA is reflected {in} both the overall MSPE and the individual MSEs for each PC. 
{In this particular data analysis task, predictive PCA has good predictive performance for the second PC but is otherwise similar to or outperformed by both classical PCA and RapPCA. More generally, in settings with high-dimensional predictors (covariates and/or spatial basis), predictive PCA is not guaranteed to show an advantage in predictability since it requires a separate processing step on these predictors instead of selecting the information to use for prediction data-adaptively as RapPCA does.} 

\begin{table}[ht]
\centering
\begin{tabular}{lllllll}
\hline
         & \multicolumn{3}{c}{Overall Metrics} & \multicolumn{3}{c}{Individual MSEs} \\
         & TMSE        & MSPE      & MSRE-trn      & PC1       & PC2        & PC3       \\ \hline
PCA      & 14.79       & 7.93     & 6.38     & 3.93       & 2.86     & 1.14     \\
PredPCA  & 14.81      & 7.16     & 7.28     & 3.41      & 2.53      & 1.22     \\
RapPCA & 13.92      & 6.66     & 6.75    & 2.62     & 2.92     & 1.12    \\ \hline
\end{tabular}
\caption{Comparison of overall metrics and individual prediction MSEs for each PC, assessed by 10-fold cross-validation on the Seattle traffic-related air pollution data. \sccomment{Standard deviations of these metrics evaluated across 10 cross-validation folds are reported in Table~\ref{tab:app-trap-cv} in Appendix~\ref{app:more-data-rslts}}.} 
\label{tab:trap-cv}
\end{table}

We then run dimension reduction with each of the PCA algorithms on the whole dataset to assess the spatial distribution of top 3 PC scores, where {for RapPCA we select} the tuning parameters {$\gamma,\lambda_1$ and $\lambda_2$ that lead} to the optimal TMSE. Next, we train spatial prediction models for each of the PC scores, and evaluate these models on a finer grid of 5,040 locations over the study region to obtain smoothed plots of each PC score.
{Figure~\ref{fig:trap-pcs} in Appendix~\ref{app:more-data-rslts} visualizes the smoothed, finer-grain PC scores obtained by each method across the study region}. Figure~\ref{fig:trap-loadings} presents the PC loadings reflecting the contribution of each pollutant on the PC scores. 
{Roughly categorizing various UFP particle sizes into small, moderate and large groups, one key difference we observe from Figure~\ref{fig:trap-loadings} between RapPCA and the two benchmark algorithms is that RapPCA identifies one size group to be the main contributor to each of the 3 PCs, while the other two groups have negligible loadings: small particles contributing to the top PC, moderately-sized particles to the second PC, and large ones to the third PC. In addition, BC is mainly contributing to the second PC according to RapPCA. 
This observation is consistent with findings from other studies, where aircraft emissions (primarily composed of small UFPs) and road traffic, especially diesel exhaust (comprising moderately-sized UFPs and BC), were identified as two key sources of traffic-related air pollution \citep{austin2021distinct, nie2022characterizing}. {These findings reflect the benefits in downstream model interpretability brought by dimension reduction that well preserves scientific and spatial information.} In contrast, neither predictive nor classical PCA effectively distinguishes the contributions of UFPs with different sizes or the contribution of BC.}

\begin{figure}
    \centering
    \includegraphics[width=16.5cm]{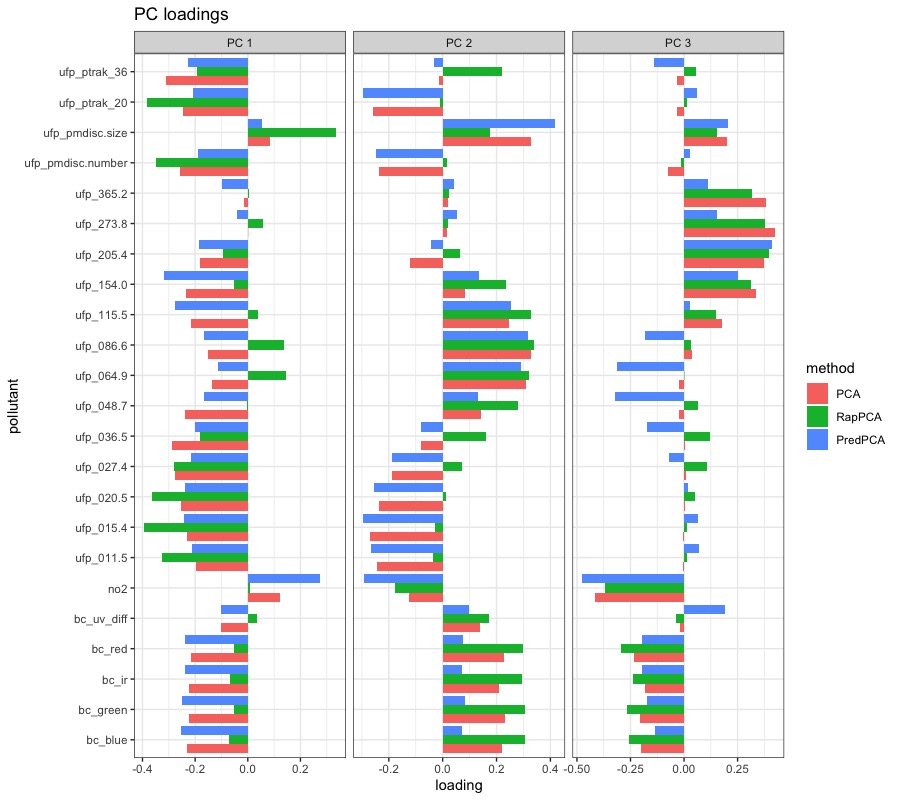}
    \caption{PC loadings {for each of the top 3 PCs obtained by PCA, predictive PCA and RapPCA} for each pollutant. There are {3} types of pollutants, {UFP, BC and NO\textsubscript{2}}. The suffix, if applicable, represents the properties of, or the instruments used to measure, each pollutant. In particular, the numeric suffix after \texttt{ufp\textunderscore} corresponds to the range of sizes for the particles. {\texttt{ufp\textunderscore ptrak\textunderscore 36} represents PTRAK measurements with diffusion screen, and \texttt{ufp\textunderscore ptrak\textunderscore 20} represents the difference between PTRAK measurements without and with diffusion screen. Black carbon (BC) measurements {corresponding to} different wavelengths are shown: blue (\texttt{bc\textunderscore blue}), green (\texttt{bc\textunderscore green}), infrared (\texttt{bc\textunderscore ir}), red (\texttt{bc\textunderscore red}), and ultraviolet (\texttt{bc\textunderscore uv}). Ultraviolet measurements were transformed to represent the difference (\texttt{bc\textunderscore uv\textunderscore diff}) between ultraviolet and infrared ranges.}}
    \label{fig:trap-loadings}
\end{figure}

\subsection{Analysis of HER2-Positive Breast Tumor {Spatial Transcriptomics}}
\label{subsec:st}

{In this section,} we present another use case of RapPCA by analyzing spatial transcriptomics data. While dimension reduction is also a common data handling step in these applications, there are a variety of downstream modeling tasks where spatial prediction, as in the Seattle TRAP study, may not be the primary focus. We will nevertheless illustrate that optimizing the balance between predictability and representability with RapPCA still leads to desirable performance in these tasks.

We analyze the human epidermal growth factor receptor 2 positive ({HER2+}) breast tumor data, which {includes} expression {measures}  of genes in HER2+ tumors from eight individuals {\citep{andersson2021spatial}}. We focus on the first tumor section of the last individual, coded as sample H1 in \citet{andersson2021spatial}, where the data consists of expression counts of $15,030$ genes on $613$ {tissue} locations. Following the same filtering steps as in \citet{shang2022spatially}, we removed genes with non-zero expression at less than 20 locations and those confounded by technical artifacts \citep{andersson2021spatial}, as well as locations with non-zero expressions for less than 20 genes. The filtered set of data contains $p=10,053$ 
genes measured on $n=607$ spots, which is then normalized via regularized negative binomial regression as implemented by the \texttt{Seurat R} package \citep{hafemeister2019normalization}.

Due to the large number of genes along with the noise present in the expression measurements \citep{rao2021exploring}, dimension reduction is commonly performed before the main analytic tasks on genomic data to extract biologically meaningful signals as well as to facilitate computation; see \citet{sun2019accuracy, zhao2021spatial, shang2022spatially} for related examples. We focus on two modeling tasks: \textit{spatial extrapolation} and \textit{domain detection}. In addition to the three PCA methods (PCA, predictive PCA, RapPCA) that were investigated in previous sections, we also include spatial PCA in our comparison. Spatial PCA is a probabilistic PCA algorithm for spatial transcriptomics data which incorporates spatial proximity information when modeling PC scores {\citep{shang2022spatially}}.

The first application is more similar to the analysis in Section~\ref{subsec:trap}, where we are interested in {predicting the gene expression PC scores} at tissue locations with {no} measurements. This prediction problem is involved, for example, in the reconstruction of high-resolution spatial maps of gene expression {\citep{gryglewski2018spatial, shang2022spatially}}. For each dimension reduction method, we extract the top 20 PCs from the expression of all genes, and conduct 10-fold cross-validation to assess TMSE, MSPE and MSRE-trn, where spatial prediction is done by random forest followed by TPRS smoothing as in previous sections. For spatial PCA, we also include results from its built-in spatial prediction function (using the \texttt{SpatialPCA R} package) for completeness. 
\sccomment{There are no covariates involved and we only use spatial information for both the spatial extrapolation and domain detection tasks. We adopt the same spatial splines and hyperparameter ranges as in the air pollution study in Section~\ref{subsec:trap}, and recall that the dimension $m$ of spatial basis is truncated to 10 for predictive PCA whereas RapPCA sets this dimension to the size of the training data. } 
Table~\ref{tab:st-cv} summarizes the overall metrics along with individual prediction MSEs for the top 3 PCs. RapPCA demonstrates significant advantage in overall prediction accuracy (MSPE) compared to all other methods, followed by predictive PCA which also has the lowest per-PC MSEs for the top 3 PCs. 
Consistent with findings from previous sections, RapPCA achieves the optimal trade-off between prediction and approximation gaps, as reflected by TMSE.

\begin{table}[]
\centering
\begin{tabular}{lllllll}
\hline
                     & \multicolumn{3}{c}{Overall Metrics} & \multicolumn{3}{c}{Individual MSEs} \\
                     & TMSE      & MSPE     & MSRE-trn      & PC1       & PC2       & PC3      \\ \hline
PCA                  & 272.21    & 32.24    & \sccomment{219.97}   & 6.64      & 5.09      & 3.18     \\
PredPCA              & 278.08    & 31.85    & \sccomment{241.91 }    & 3.93      & 4.12      & 2.26    \\
SpatialPCA          & 292.78     & 46.09    & \sccomment{236.88}     & 6.84      & 5.89     & 3.69     \\
SpatialPCA: built-in & 296.59    & 49.90    & \sccomment{236.88 }     & 6.81      & 5.96      & 3.81     \\
RapPCA               & 271.35    & 20.56   & \sccomment{243.27}     & 5.59       & 4.86      & 3.79     \\ \hline
\end{tabular}
\caption{Comparison of overall metrics and individual prediction MSEs for each of the top 3 PCs, assessed by 10-fold cross-validation on the breast tumor data. \sccomment{Standard deviations of these metrics evaluated across 10 cross-validation folds are reported in Table~\ref{tab:app-st-cv} in Appendix~\ref{app:more-data-rslts}}.}
\label{tab:st-cv}
\end{table}

Figure~\ref{fig:st-pcs} visualizes the top 3 PC scores from the expressions of all genes. We observe that spatial PCA and predictive PCA produce spatially smoother surfaces of PC scores compared to the other two methods. This is natural for predictive PCA since it prioritizes spatial predictability of the PC scores, which results in stronger spatial smoothing; for spatial PCA, the prediction-representation balance is implicit and less straightforward to interpret. The gap in the prediction performance of spatial or predictive PCA {compared} with RapPCA {may be} the consequence of over-smoothing, which also indicates the effectiveness of explicit optimization for the prediction-representation balance achieved by (\ref{equ:rappca}). {Figure~\ref{fig:st-highres} in Appendix~\ref{app:more-data-rslts} visualizes the smoothed high-resolution maps of PC scores based on predictions at unmeasured locations.}

\begin{figure}[!ht]
    \centering
    \includegraphics[width=16.5cm]{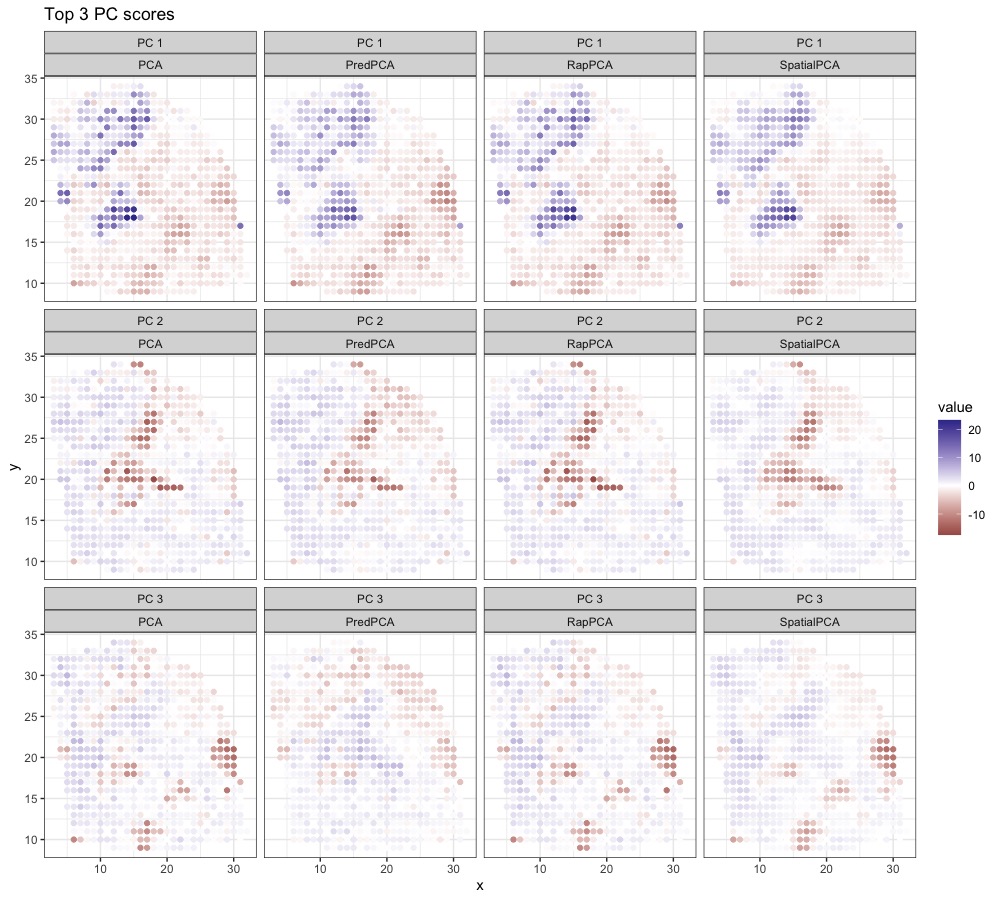}
    \caption{Spatial distribution of the top 3 PC score values {extracted by PCA, predictive PCA, RapPCA and spatial PCA}, based on measured gene expression in the HER2-positive breast tumor data.}
    \label{fig:st-pcs}
\end{figure}

In the second application, we investigate the problem of domain detection following the dimension reduction step. In domain detection, different sections of tissues are identified as various spatially coherent and functionally distinct regions \citep{dong2022deciphering, shang2022spatially}, {providing} helpful insights {into} the biological function of tissues. To this end, we conduct domain detection via the walktrap clustering algorithm \citep{pons2005computing} using the top 20 PCs extracted by each {algorithm}. The number of clusters is set to align with the ``ground truth'' labels based on pathologist annotations of tissue regions in the original study \citep{andersson2021spatial}. {The hyperparameters $\gamma,\lambda_1, \lambda_2$ are chosen based on silhouette value measuring the difference between nearest-cluster versus intra-cluster distances \citep{rousseeuw1987silhouettes}.}

\sccomment{To provide some heuristics on our choice of $r$ as described in Section~\ref{subsec:method-2}, Figure~\ref{fig:elbow_curves} visualizes the trends for the cumulative sum of PC-wise prediction MSEs, along with MSRE-trn of RapPCA. We observe that the performance of RapPCA starts to stabilize for $r$ between 10 and 15, justifying $r=20$ to be sufficient.}

\begin{figure}
    \centering
    \includegraphics[width=15cm]{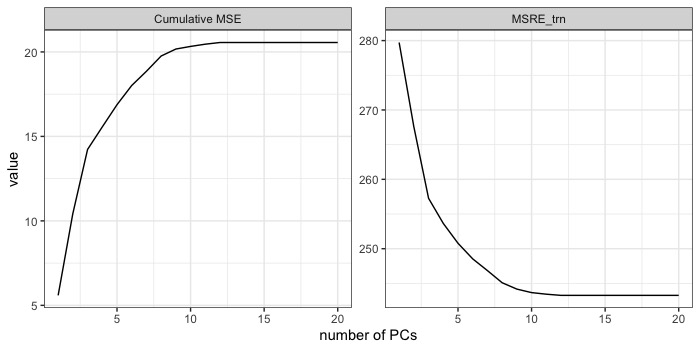}
    \caption{\sccomment{Changes in cumulative prediction MSE and MSRE-trn with the number of PCs for RapPCA on the breast tumor dataset. 
    }}
    \label{fig:elbow_curves}
\end{figure}

Figure~\ref{fig:st-domains} visualizes the detected spatial domains based on PCs extracted from different dimension reduction procedures, compared to the ground truth labels. In general, the relative performance of each algorithm differs across spatial domains. For example, PCA fails to detect the adipose tissue region and produces noisier results for regions surrounding invasive cancer cells. Predictive PCA and RapPCA both mis-classify part of the cancer in situ region as invasive cancer, where predictive PCA results are noisier for the top left region, and RapPCA shows larger uncertainty regarding adipose tissue. Spatial PCA, on the other hand, infers part of the connective tissue to be cancerous. 

\begin{figure}[!ht]
    \centering
    \includegraphics[width=13cm]{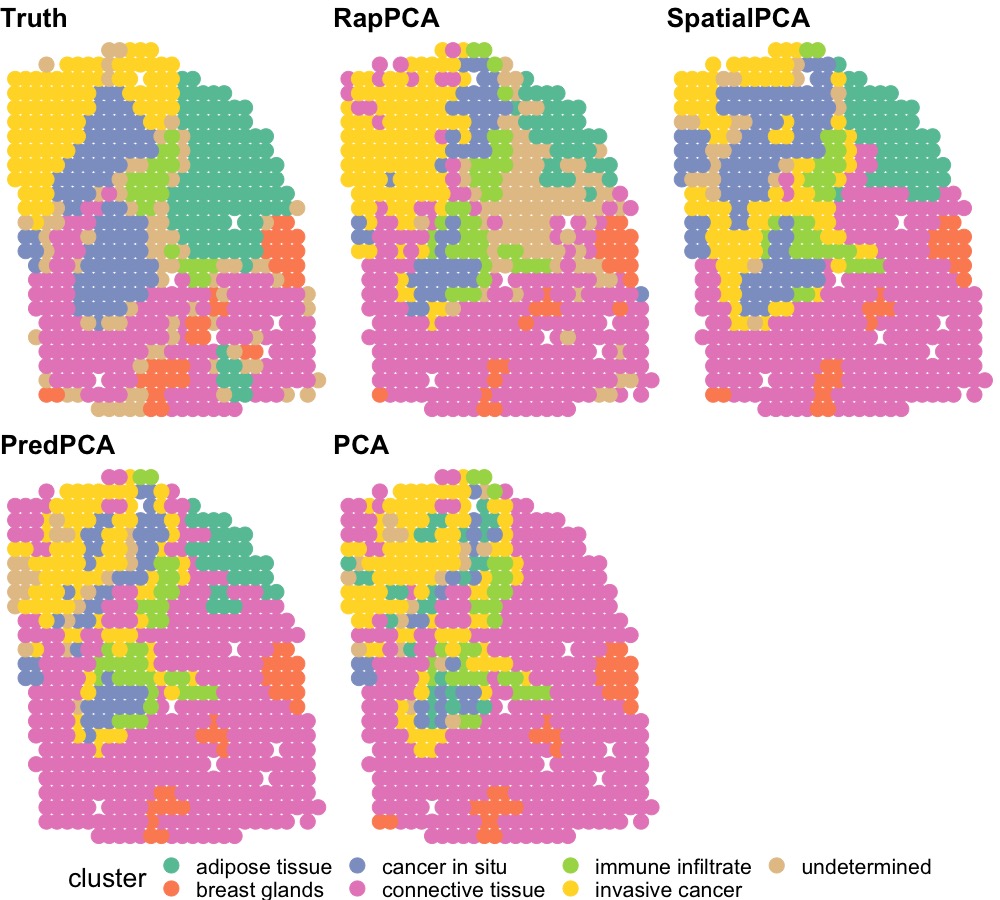}
    \caption{Detected spatial domains and annotated ground truth}
    \label{fig:st-domains}
\end{figure}

To make a more in-depth comparison, we examine the breakdown of clustering accuracy by ground truth labels; {that is}, whether or not each approach correctly classifies the spots belonging to each spatial domain, with the undetermined cluster removed from comparison. Figure~\ref{fig:st-metrics} presents this breakdown for each method in terms of precision, recall and F1 score. As observed from the spatial domain plots, the relative accuracy for each method differs by region, and in particular, RapPCA has the best performance in both precision and recall on the identification of invasive cancer regions, whereas spatial PCA shows an advantage in recall for cancer in situ. 

\begin{figure}
    \centering
    \includegraphics[width=16.5cm]{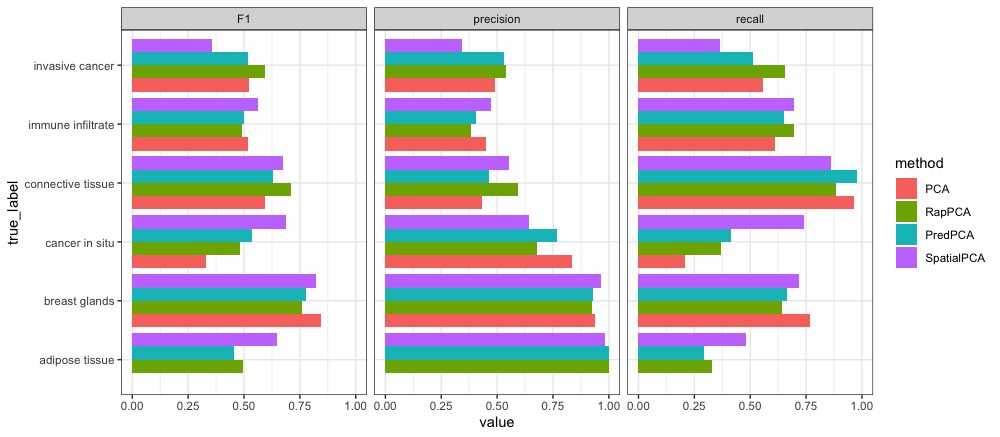}
    \caption{Breakdown of domain detection accuracy by true label for each algorithm. Note that the metrics for the adipose tissue region {(bottom row)} are not well-defined for PCA because it fails to detect any spot in this region.}
    \label{fig:st-metrics}
\end{figure}
\section{Discussion}
\label{sec:discussion}

When dimension reduction is conducted as an intermediate step before a spatial modeling task of interest, there are two typically conflicting considerations guiding the choice of approaches. The more obvious one is how closely the lower dimensional score represents the original variables, which we refer to as \textit{representability}. Another important aspect is the \textit{predictability} of the resulting scores, i.e., how well they could be expressed or predicted by {available covariates and/or spatial effects}. 

We discussed how existing dimension reduction algorithms fit into this general framework of optimizing either criterion, and proposed a flexible interpolation between them that achieves the optimal representability-predictability trade-off. {Our proposal, called RapPCA, can also handle} high-dimensional {predictors (including covariates and spatial basis)} and non-linear relationships {between covariates and PC scores}. Simulation studies under different scenarios illustrate the gain in downstream prediction accuracy by our method, which also achieves smaller overall errors when both predictability and representability are taken into account.
Applications to real datasets in different scientific domains further demonstrate the utility of our proposed method, even in analytic tasks that do not explicitly involve prediction.

The balance between prediction and representation performance can be viewed as a generalized form of bias-variance trade-off. While methods such as {classical} PCA minimize the representation gap specifically for the training data at hand, such representation may capture excessive noise if it explains a large amount of variability in the data. In contrast, the formulation of predictive PCA restricts the PC scores to fall within a certain model space, the linear span of the covariates, which is a form of regularization enforcing the smoothness of the PCs. Many probabilistic dimension reduction approaches implicitly address both aspects, while our proposal seeks the optimal balance in an explicit and interpretable way.

While we used TMSE to guide the choice of tuning parameters in {most of} our examples, {in practice tuning parameter selection could} be driven by the specific analytic goal. {For example, we used silhouette value reflecting the goodness of clustering for the domain detection task in Section~\ref{subsec:st}; while for the health effect analysis of air pollution in \citet{jandarov2017novel}, the main focus was spatial extrapolation accuracy and hence prediction error is of primary interest}. Researchers could examine the trend of prediction, representation and total errors by each tuning parameter, and in particular $\gamma$ in (\ref{equ:rappca}), on a set of test data, and choose the combination leading to the desired trade-off. {Our empirical evaluations suggest} that striking such balance with our proposed method would most often improve, or at least maintain similar, overall errors compared to the existing alternatives.


\bibliographystyle{rss}
\bibliography{ref}

\newpage
\newpage

\begin{center}
{\large\bf APPENDIX}
\end{center}

\appendix
\renewcommand{\thefigure}{S\arabic{figure}} 
\setcounter{figure}{0} 
\renewcommand{\theHfigure}{S\arabic{figure}}
\renewcommand{\thetable}{S\arabic{table}} 
\setcounter{table}{0} 
\renewcommand{\theHtable}{S\arabic{figure}}
\section{Proof}
\label{app:proof}
\begin{proof}[Proof of Theorem~\ref{thm:optim}]

    {For the $l$th component}, denote the singular value decomposition of $Y^{(l)}$ as $Y^{(l)} = \mathcal S^{(l)}{\mathcal D^{(l)}} {\mathcal T^{(l)}}^\top$. Let ${\eta}:=\left[\alpha^\top \quad \sqrt{\frac{\lambda_2}{\lambda_1}}\beta^\top\right]^\top$, $Z := \left[K \quad \sqrt{\frac{\lambda_2}{\lambda_1}}B\right]$, and let the combined penalty matrix be
    \[
    P := \begin{bmatrix}
        \tilde K & 0 \\ 0 & \frac{\lambda_2}{\lambda_1}\tilde Q
    \end{bmatrix}.
    \]
    {We verify that} the optimization problem (\ref{equ:rappca}) has a unique {global minimizer} given by
        \begin{itemize}
            \item {$\tilde v$ as defined in Algorithm~\ref{thm:optim}};
            \item $\tilde \eta = (\gamma Z^\top Z + \lambda_1 P)^{-1}(\gamma Z^\top \mathcal S^{(l)}{\mathcal D^{(l)}}q)$, with $q$ defined {as in Algorithm~\ref{thm:optim}};
        \end{itemize}
    {by doing so, we  establish the optimality of our solution given by Algorithm~\ref{algo:rappca}.}
    
    Recall the form of the optimization problem given in (\ref{equ:rappca}):
    \begin{align*}
        \notag &\min_{v,\alpha,\beta}\ f_{\gamma,\lambda_1,\lambda_2}(v,\alpha,\beta):= \left\lVert Y^{(l)}-Y^{(l)}vv^\top \right\rVert_F^2 +
    \gamma\left\lVert Y^{(l)}v-(K\alpha+B\beta) \right\rVert_2^2 +\lambda_1 \alpha^\top \tilde K\alpha + \lambda_2 \beta^\top \tilde Q\beta \\
    &\text{s.t. } v^\top v = 1.
    \end{align*}
    {Suppressing the superscript $(l)$ for simplicity,} we denote $y = Y^{(l)}$ and start by examining the first term in the objective function, which can be expressed as
    \begin{align}
        \notag \left\lVert Y^{(l)}-Y^{(l)}vv^\top \right\rVert_F^2
        &= \text{tr}\left( (y - yvv^\top)^\top (y-yvv^\top) \right) \\
        \notag &= \text{tr}\left(y^\top y - 2y^\top yvv^\top + vv^\top y^\top yvv^\top\right) \\
        &\propto -2\text{tr}\left( v^\top {\mathcal T}\mathcal D\mathcal S^\top \mathcal S\mathcal D{\mathcal T}^\top v \right) + \trace{v^\top vv^\top {\mathcal T}\mathcal D{\mathcal S}^\top {\mathcal S}\mathcal D{\mathcal T}^\top v},
        \label{equ:term1-trace}
    \end{align}
    where $\text{tr}(\cdot)$ denotes the trace of a matrix, $\propto$ means equal up to a constant not depending on $v,\alpha$ or $\beta$, and we have made use of the cyclic property of trace. Since $v^\top v = 1$ and ${\mathcal S}^\top {\mathcal S} = I$, denoting $q:={\mathcal T}^\top v$, (\ref{equ:term1-trace}) {can be written} as
    \begin{align*}
        \left\lVert Y^{(l)}-Y^{(l)}vv^\top \right\rVert_F^2
        &= -2\trace{q^\top \mathcal D^2 q} + \trace{q^\top \mathcal D^2 q} = -{q^\top  \mathcal D^2 q}.
    \end{align*}

    Under the reparametrization in Theorem~\ref{thm:optim}, namely, ${\eta}:=\left[\alpha^\top \quad \sqrt{\frac{\lambda_2}{\lambda_1}}\beta^\top\right]^\top$ and $Z := \left[K \quad \sqrt{\frac{\lambda_2}{\lambda_1}}B\right]$, the objective function becomes
    \begin{align}
        \notag f_{\gamma,\lambda_1,\lambda_2}(v,\eta) 
        &= -q^\top \mathcal D^2 q + \gamma\left\lVert {\mathcal S}\mathcal Dq-Z\eta \right\rVert_2^2 +\lambda_1 \eta^\top P \eta,
    \end{align}
    which is a quadratic function of $\eta$ for fixed $q$, with
    \[
        \frac{\partial f}{\partial \eta} = 2\gamma Z^\top (Z\eta - {\mathcal S}\mathcal Dq) + 2\lambda_1 \sccomment{P^{1/2}}\eta,
    \]
    \sccomment{where the square root $P^{1/2}$ is well-defined because $\tilde K$ and $\tilde Q$ are both positive semidefinite.} Since the Hessian $\partial^2 f/\partial\eta^2 = 2\gamma Z^\top Z + \sccomment{2\lambda_1P^{1/2}}$ is positive semidefinite, we can profile out $\eta$ by setting it to the minimizer $\tilde\eta(q) := (\gamma Z^\top Z + \lambda_1 P)^{-1}(\gamma Z^\top {\mathcal S}\mathcal Dq)$. The objective function can thus be rearranged as
    \begin{align}
        \notag f_{\gamma,\lambda_1,\lambda_2}(v,\alpha,\beta)
        &:= g_{\gamma,\lambda_1}(q,\tilde\eta) 
        = -q^\top \mathcal D^2 q + \gamma({\mathcal S}\mathcal Dq - Z\tilde\eta)^\top ({\mathcal S}\mathcal Dq - Z\tilde\eta) + \lambda_1 \tilde\eta^\top P \tilde\eta \\
        \notag &= -q^\top \mathcal D^2 q + \gamma(q^\top \mathcal D^2 q - 2\tilde\eta^\top Z^\top {\mathcal S}\mathcal Dq + \tilde\eta^\top Z^\top Z\tilde\eta) + \lambda_1\tilde\eta^\top P\tilde\eta\\
        \notag &= (\gamma - 1) q^\top \mathcal D^2 q
        - 2\gamma^2 q^\top \mathcal D{\mathcal S}^\top Z(\gamma Z^\top Z + \lambda_1 P)^{-1} Z^\top {\mathcal S}\mathcal Dq \\
        \notag &\quad\quad + \gamma^2 q^\top \mathcal D{\mathcal S}^\top Z(\gamma Z^\top Z+\lambda_1 P)^{-1}(\gamma Z^\top Z+ \lambda_1 P)(\gamma Z^\top Z+\lambda_1 P)^{-1} Z^\top {\mathcal S}\mathcal Dq \\
        &= -q^\top \left[-(\gamma-1)\mathcal D^2 + \gamma^2 \mathcal D{\mathcal S}^\top Z(\gamma Z^\top Z+\lambda_1 P)^{-1}Z^\top {\mathcal S}\mathcal D\right] q:= -q^\top A q,
        \label{equ:matrix-A}
    \end{align}
    where $A:=-(\gamma-1)\mathcal D^2 + \gamma^2 \mathcal D{\mathcal S}^\top Z(\gamma Z^\top Z+\lambda_1 P)^{-1}Z^\top {\mathcal S}\mathcal D$. 

    It then follows that the original optimization problem (\ref{equ:rappca}) is equivalent to
    \begin{align*}
        \notag &\min_{q}\ f_{\gamma,\lambda_1}(q):= -q^\top A q \quad \text{s.t. } q^\top q = 1,
    \end{align*}
    for which the optimal solution is the normalized first eigenvector of $A$ (i.e., the one corresponding to the largest eigenvalue). 
    {By construction and the properties of eigendecomposition, such $q$ is guaranteed to be the global minimizer of $f_{\gamma,\lambda_1}(q)$}.
    Recalling the form of $\tilde\eta(q)$ and the fact that the untransformed optimal solution satisfies ${\mathcal T}^\top \tilde v = \tilde q$ and $\tilde u = Y^{(l)}\tilde v$ completes the proof.
\end{proof}

\section{Additional Numerical Results}
\label{app:more-rslts}

\subsection{Optimality of Proposed Solution}
\label{app:optimality}
We use the first replicate (out of 100 total) of data in Simulation Scenario 1 (Section~\ref{subsec:sim-case1}) to numerically verify our claim in Theorem~\ref{thm:optim}, i.e., the optimality of the derived solution. 

Specifically, we solve the optimization problem (\ref{equ:rappca}) to extract the first PC score, and vary the first two entries of the optimal loadings $\tilde v$ while keeping everything else (including the PC scores $\tilde u$ and coefficients $\tilde\eta$) intact, and ensuring that the constraint $v^\top v$ is still satisfied. This is done under polar coordinates: letting $\theta\in[0, 2\pi)$, we write $v^*(\theta)$ as the modified loading vector, where the entries $v_j^*(\theta) = \tilde v_j$ for $j > 2$. We let $v^*_1(\theta) = \rho\sin\theta$ and $v^*_2(\theta) = \rho\cos\theta$ where $\rho = \sqrt{1-\sum_{j>2} {\tilde v_j}^2}$. Then, as $\theta$ varies within $[0,2\pi)$, the two entries $(v_1^*(\theta), v_2^*(\theta))$ take all possible combinations that satisfy the constraint. 

We explore different values of tuning parameters $\gamma,\lambda_1, \lambda_2$ and compare the objective function $f_{\gamma,\lambda_1,\lambda_2}(u,v,\eta)$ evaluated at the optimal solution $(\tilde u, \tilde v, \tilde\eta)$ versus the modified solution $(\tilde u, v^*(\theta), \tilde\eta)$ for $\theta\in[0,2\pi)$. We plot their differences against $\theta$ for various values of tuning parameters in Figures~\ref{fig:suppl-1} and~\ref{fig:suppl-2}. {Note that the optimal $\theta$, i.e., the value leading to the optimal loading $v^*(\theta)$, may differ with different sets of tuning parameters $\gamma,\lambda_1$ and $\lambda_2$.} The fact that each curve is always above or equal to zero, and reaches zero exactly once, verifies the uniqueness of the optimal solution, as well as the optimality of the proposed form of $\tilde v$, at least for fixed $\tilde u$ and $\tilde \eta$.

\begin{figure}[ht]
    \centering
    \includegraphics[width=16.5cm]{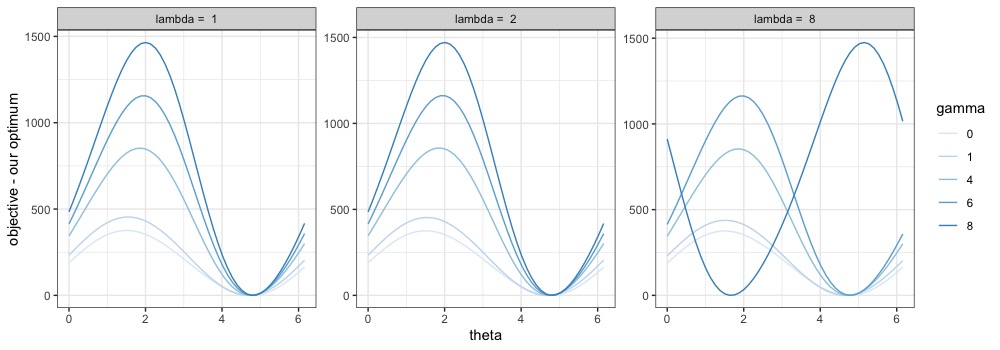}
    \caption{{Difference between the objective function evaluated over a range of values for the first two parameters in the loading vector and the optimum achieved by our algorithm. The $x$-axis is parameterized by $\theta$, the polar coordinate angular representation of the candidate parameters, and we fix $\lambda_1=\lambda_2$. Each panel corresponds to a value of $\lambda_1$ and $\lambda_2$.}}
    \label{fig:suppl-1}
\end{figure}

\begin{figure}[ht]
    \centering
    \includegraphics[width=16.5cm]{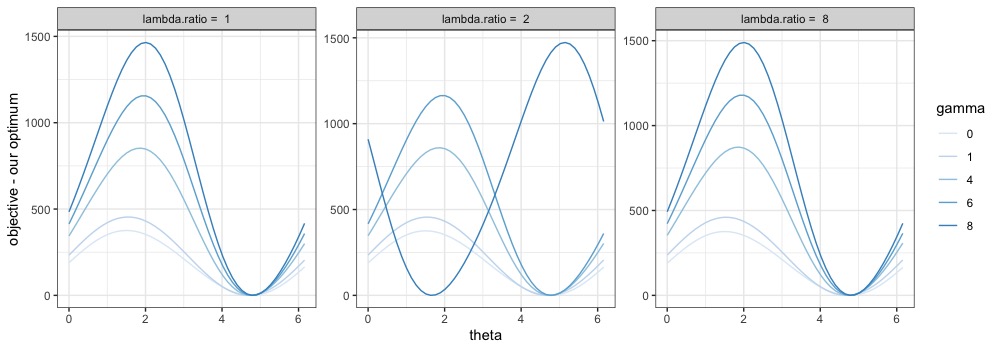}
    \caption{{Difference between the objective function evaluated over a range of values for the first two parameters in the loading vector and the optimum achieved by our algorithm. The $x$-axis is parameterized by $\theta$, the polar coordinate angular representation of the candidate parameters, and we fix $\lambda_1=1$. Each panel corresponds to a value of the ratio $\lambda_2/\lambda_1$.}}
    \label{fig:suppl-2}
\end{figure}

\subsection{Additional Simulation Results}
\label{app:more-sim-rslts}

{Figure~\ref{fig:sim-lambdas} is an analog of Figure~\ref{fig:sims-gm} presented in our simulation studies, which visualizes the variation of TMSE, MSPE and MSRE-trn for the first PC with respect to the combination of $\lambda_1$ and $\lambda_2$, instead of $\gamma$. For each combination, the additional tuning parameter $\gamma$ is chosen such that it leads to the lowest TMSE.
In Scenarios 1 through 3, the median values of the selected parameters $(\lambda_1, \lambda_2/\lambda_1)$ via cross-validation are (0.5, 1), (1, 0.75) and (0.5, 0.25) respectively, all of which are reasonably close to near-optimal region(s) in Figure~\ref{fig:sim-lambdas}. This reflects the effectiveness of the cross-validation procedure to control the degree of freedom of coefficients $\alpha$ and $\beta$ in (\ref{equ:rappca}).} 

\subsection{Additional Data Analysis Results}
\label{app:more-data-rslts}

\sccomment{Tables~\ref{tab:app-trap-cv} and~\ref{tab:app-st-cv} report the standard deviations of all evaluation metrics in Tables~\ref{tab:trap-cv} and~\ref{tab:st-cv} respectively, capturing their variability across 10 cross-validation folds. } 

\sccomment{The larger variability of MSRE-trn for RapPCA is expected. We recall that RapPCA includes classical PCA as a special case; and therefore for each PC, RapPCA may revert to classical PCA and have smaller MSRE-trn in some cross-validation folds than others. Nevertheless, the comparable variability of TMSE between RapPCA and other methods indicate that the overall performance and predictability-representability tradeoff of RapPCA remain consistently superior to other benchmark methods.}

\begin{table}[!ht]
\centering
\begin{tabular}{lllllll}
\hline
         & \multicolumn{3}{c}{Overall Metrics} & \multicolumn{3}{c}{Individual MSEs} \\
         & TMSE        & MSPE      & MSRE-trn      & PC1       & PC2        & PC3       \\ \hline
PCA      & 4.45       & 3.50     & 0.16     & 2.24       & 1.60      & 0.52     \\
PredPCA  & 4.80      & 3.22     & 0.30     & 1.95      & 1.49      & 0.39      \\
RapPCA & 4.26      & 2.73     & 0.65     & 2.25      & 1.51      & 0.69     \\ \hline
\end{tabular}
\caption{\sccomment{Standard deviations for all evaluation metrics across 10 cross-validation folds for the Seattle traffic-related air pollution data.}}
\label{tab:app-trap-cv}
\end{table}

\begin{table}[]
\centering
\begin{tabular}{lllllll}
\hline
                     & \multicolumn{3}{c}{Overall Metrics} & \multicolumn{3}{c}{Individual MSEs} \\
                     & TMSE      & MSPE     & MSRE-trn      & PC1       & PC2       & PC3      \\ \hline
PCA                  & 10.86    & 3.66    & \sccomment{0.90}   & 2.22      & 1.79      & 0.76     \\
PredPCA              & 11.75    & 2.96    & \sccomment{1.39}    & 1.06     & 1.13      & 0.74     \\
SpatialPCA          & 11.45    & 4.74    & \sccomment{1.21}     & 2.24      & 2.13      & 0.90     \\
SpatialPCA: built-in & 12.04    & 5.22    & \sccomment{1.21}     & 2.18      & 2.06      & 0.96     \\
RapPCA               & 10.82    & 4.11    & \sccomment{6.04}     & 2.21      & 1.59      & 1.16     \\ \hline
\end{tabular}
\caption{\sccomment{Standard deviations for all evaluation metrics across 10 cross-validation folds for the breast tumor data.}}
\label{tab:app-st-cv}
\end{table}

\begin{figure}
    \centering
    \includegraphics[width=16.5cm]{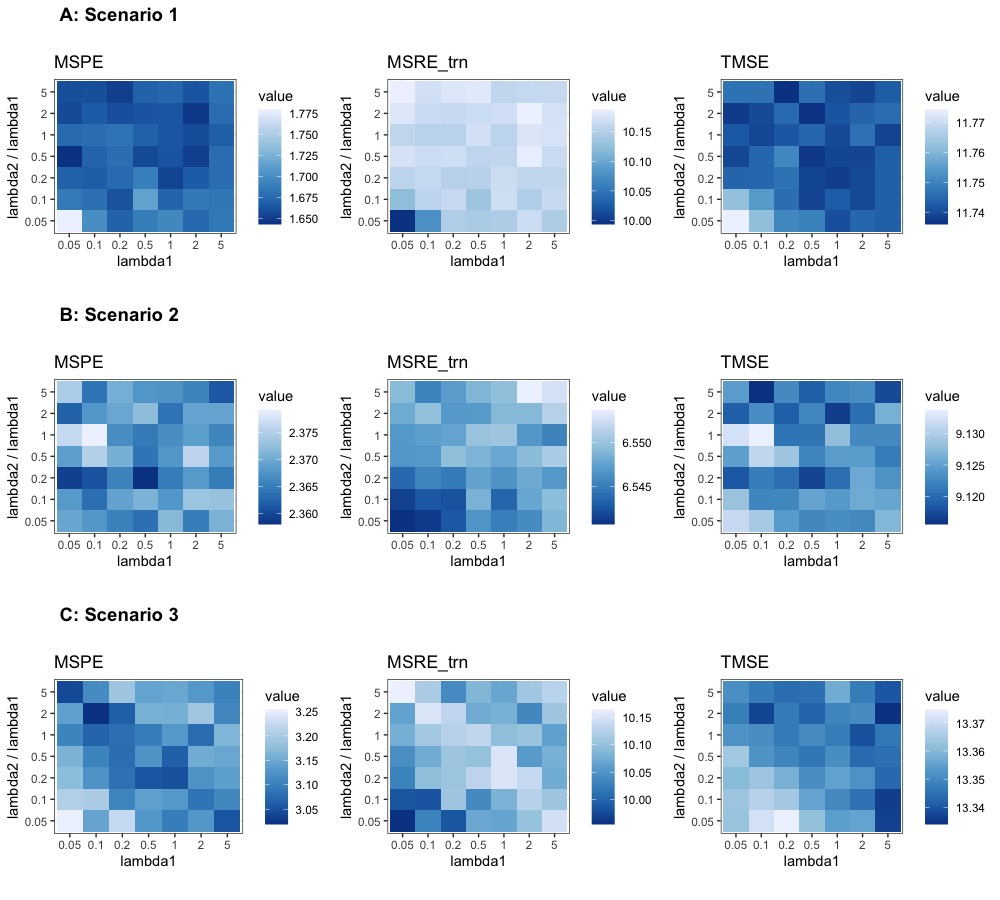}
    \caption{{Average TMSE, MSPE and MSRE-trn for the first PC extracted by RapPCA arcoss 100 replicates of data, plotted against the combination of $\lambda_1$ and the ratio $\lambda_2/\lambda_1$.}}
    \label{fig:sim-lambdas}
\end{figure}

{For the TRAP data analysis presented in Section~\ref{subsec:trap}, Figure~\ref{fig:trap-pcs} presents the smoothed, finer-grain PC scores obtained by each method across the study region. We observe similar spatial patterns in the distribution of the PC scores {across different methods}, except for the south end of the study region for the third PC where PCA identifies a stronger signal than RapPCA and predictive PCA. In particular, all methods highlight regions near the airport and/or around major roads (the dark red area) for the first PC, indicating aircraft and road traffic emissions as a major source of overall pollution level. This is consistent with the large contributions of BC as well as UFP with small or moderate sizes \citep{zhang2019number, bendtsen2021review}, as reflected by the loadings of UFP (with relatively small particle sizes) and BC in Figure~\ref{fig:trap-loadings}.}

\begin{figure}
    \centering
    \includegraphics[width=14.5cm]{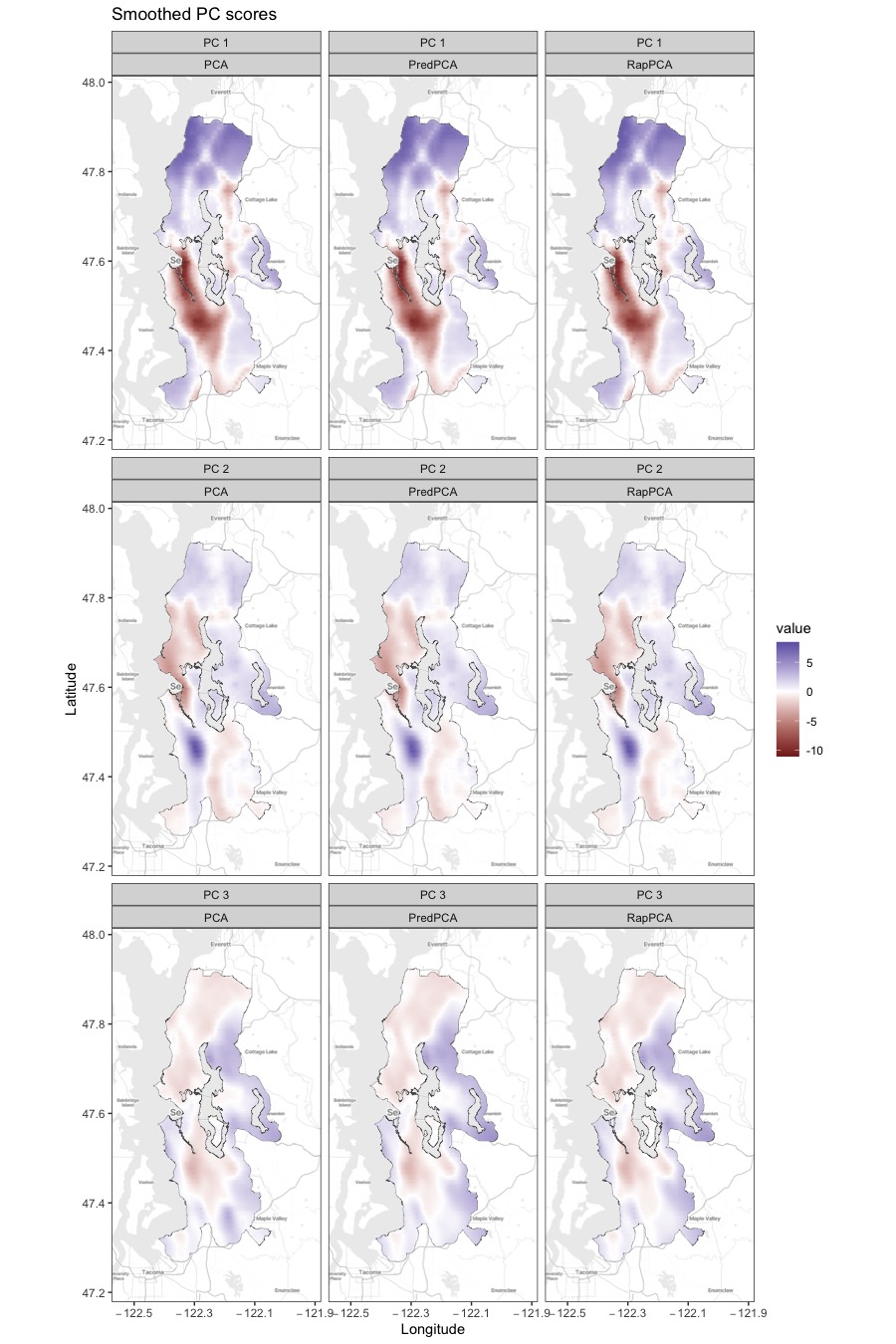}
    \caption{Smoothed PC scores of pollutant concentrations from the Seattle TRAP data, {based on spatial extrapolation following each dimension reduction algorithm.}}
    \label{fig:trap-pcs}
\end{figure}

For the spatial transcriptomics application in Section~\ref{subsec:st}, {Figure~\ref{fig:st-highres} presents the smoothed high-resolution maps of predicted PC scores, obtained by spatial prediction following each dimension reduction approach. We observe that predictive PCA produces the smoothest predicted map due to its emphasis on the predictability of PC scores; however, it could over-smooth and omit meaningful spatial variations in gene expression, as can be seen in the high-resolution map for the third PC. For spatial PCA, while the extracted low-resolution PCs also demonstrate smoothness as predictive PCA (see Figure~\ref{fig:st-pcs}), they are not guaranteed to be well predictable and the extrapolated high-resolution map appears noisier. RapPCA achieves a reasonable balance among all algorithms.}

\begin{figure}
    \centering
    \includegraphics[width=16.5cm]{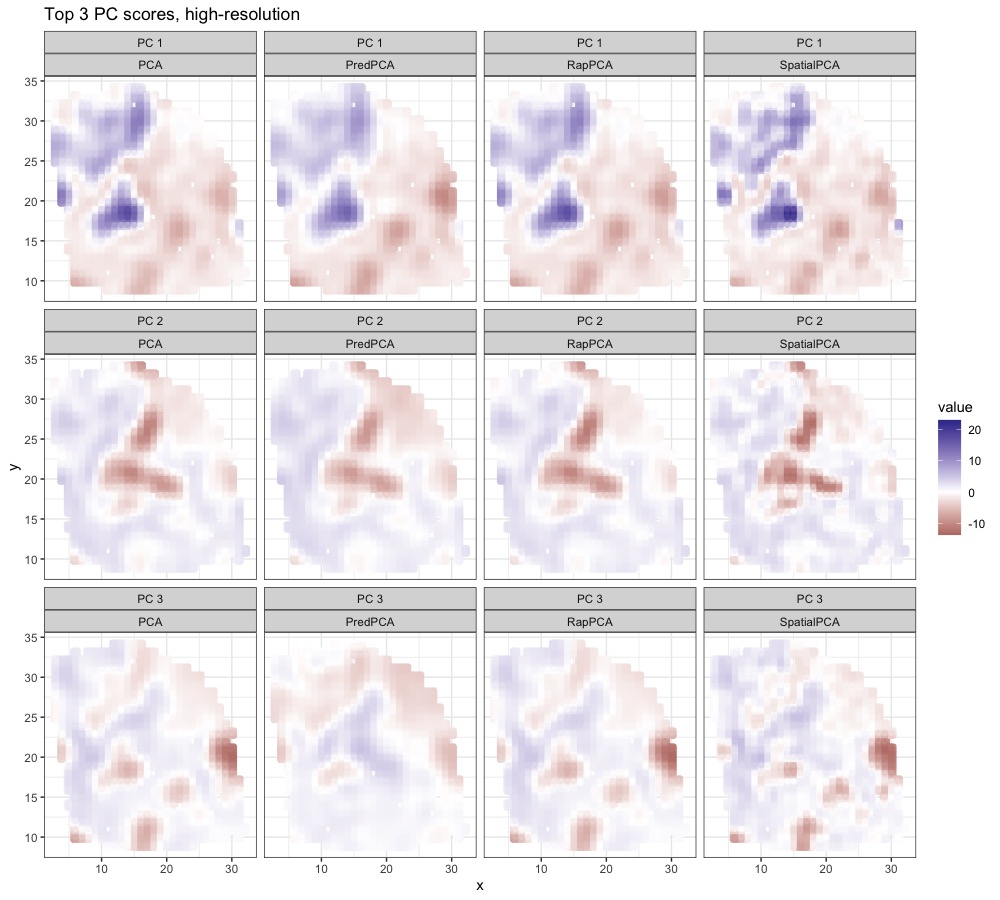}
    \caption{{Smoothed PC scores based on predicted high-resolution gene expression maps following each dimension reduction algorithm in the HER2-positive breast tumor data.}}
    \label{fig:st-highres}
\end{figure}

\end{document}